\documentclass[11pt,a4paper]{article}
\usepackage{amssymb}
\usepackage{latexsym}
\usepackage{amsmath}
\usepackage{graphicx}
\usepackage{amsthm}
\usepackage{empheq}
\usepackage{bm}
\usepackage{booktabs}
\usepackage[dvipsnames]{xcolor}
\usepackage{pagecolor}
\usepackage{subcaption}
\usepackage{tikz,pgfplots,lipsum,lmodern}
\usepackage{enumitem}
\usepackage{changepage}
\usetikzlibrary{positioning}
\usetikzlibrary{arrows.meta}
\usepackage[most]{tcolorbox}
\usepackage{appendix}
\usepackage{thm-restate}

\usepackage[margin=2.7cm]{geometry}

\usepackage{hyperref}
\usepackage[capitalize]{cleveref}

\numberwithin{equation}{section}

\theoremstyle{definition}
\newtheorem{theorem}{Theorem}[section]
\newtheorem{corollary}[theorem]{Corollary}
\newtheorem{proposition}[theorem]{Proposition}
\newtheorem{definition}[theorem]{Definition}
\newtheorem{example}[theorem]{Example}

\newtheorem{notation}[theorem]{Notation}
\newtheorem{remark}[theorem]{Remark}
\newtheorem{lemma}[theorem]{Lemma}
\newtheorem{conjecture}[theorem]{Conjecture}

\newtheorem{problem}[theorem]{Problem}

\usepackage{calrsfs}

\newcommand{\numberset}{\mathbb}

\newcommand{\F}{\numberset{F}}
\newcommand{\E}{\numberset{E}}

\newcommand{\fq}{\F_q}

\newcommand{\mC}{\mathcal{C}}

\newcommand{\mG}{\mathcal{G}}

\newcommand{\mF}{\mathcal{F}}

\newcommand{\mN}{\mathcal{N}}

\newtheorem{claim}{Claim}

\newcommand*{\myproofname}{Proof of the claim}
\newenvironment{clproof}[1][\myproofname]{\begin{proof}[#1]}{\end{proof}}

\def\BibTeX{{\rm B\kern-.05em{\sc i\kern-.025em b}\kern-.08em
    T\kern-.1667em\lower.7ex\hbox{E}\kern-.125emX}}

\title{The Geometry of Codes for Random Access in DNA Storage}
\usepackage{authblk}

\author[1]{Anina Gruica}
\affil[1]{Technical University of Denmark, Denmark}

\author[1]{Maria Montanucci}

\author[2]{Ferdinando Zullo}
\affil[2]{Università degli Studi della Campania ``Luigi Vanvitelli'', Italy}

\date{}

\usepackage{setspace}
\begin{document}\maketitle
\setcounter{MaxMatrixCols}{20}

\begin{abstract}
Effective and reliable data retrieval is critical for the feasibility of DNA storage, and the development of random access efficiency plays a key role in its practicality and reliability. In this paper, we study the Random Access Problem, which asks to compute the expected number of samples one needs in order to recover an information strand. Unlike previous work, we took a geometric approach to the problem, aiming to understand which geometric structures lead to codes that perform well in terms of reducing the random access expectation (Balanced Quasi-Arcs). As a consequence, two main results are obtained. The first is a construction for $k=3$ that outperforms previous constructions aiming to reduce the random access expectation. The second, exploiting a result from~\cite{gruica2024reducing}, is the proof of a conjecture from~\cite{bar2023cover} for rate $1/2$ codes in any dimension.
\end{abstract}

\medskip

\section{Introduction}

As the amount of data produced globally continues to rise exponentially, the demand for efficient, durable, and scalable storage solutions has become critical~\cite{rydning2022worldwide}. Traditional storage technologies struggle to keep up with this growth, both in terms of physical space and energy efficiency. DNA storage has emerged as a promising alternative~\cite{anavy2019data,blawat2016forward,bornholt2016dna, organick2018random, yazdi2017portable, tabatabaei2015rewritable, bar2021deep}, attracting significant attention due to its potential for extreme data density and long-term stability~\cite{alliance2021preserving, markowitz2023biology}. This alternative approach to data storage has attracted growing interest across fields such as biology, chemistry, computer science, electrical engineering and also mathematics. However, despite its promise, DNA storage still faces significant challenges and presents many open questions~\cite{shomorony2022information, yazdi2015dna, alliance2021preserving}.

The DNA storage process generally involves three main steps: DNA synthesis, storage, and sequencing.
Digital data, typically represented as a string of bits, is first encoded into sequences over the DNA alphabet $\{A,C,G,T\}$. This encoded string is divided into blocks, and then synthesized into actual DNA strands. These strands are stored in a container, where each strand exists in numerous copies, but in a completely unordered manner. When a user wants to access the stored data, sequencing technology retrieves a multiset of ``reads'' -- copies of the strands that may contain errors. By obtaining enough reads, users can reconstruct the original data. The number of reads required for accurate data retrieval is called \textit{coverage depth}, and reducing this depth is crucial for lowering sequencing costs and improving retrieval speed~\cite{heckel2019characterization}. To reduce the cost of DNA synthesis and sequencing, which is one of the main drawbacks of using synthesized DNA as a storage solution, error-correcting codes have been used to make the process more efficient, through reducing the number of strands needed for data retrieval.

In this paper, we focus on a specific problem arising in DNA storage systems known as the Random Access Problem, initially introduced in~\cite{bar2023cover}. This problem concerns the scenario in which a user wants to retrieve only one specific information strand from a storage system. Assuming error-free synthesis and sequencing, we aim to study and determine error-correcting codes that reduce the expected number of reads necessary to access this target strand. To achieve this, we build on results from both~\cite{bar2023cover} and~\cite{gruica2024reducing} to gain a deeper understanding of which properties of codes play a crucial role in reducing the expected number of reads required for reliable data retrieval. In our approach, we look at this problem from a geometric perspective, which provides new insights into what properties of codes play a role for them to perform well in terms of the random access problem in DNA storage systems. Moreover, studying this problem from this new viewpoint allows us to derive clean formulas for computing the random access expectation, for certain classes of codes. In contrast, previous work often involved formulas that were challenging to handle or analyze explicitly.

This paper is organized as follows. Section~\ref{sec:prelim} provides the definitions used throughout the paper, formally introduces the Random Access Problem, and offers a brief overview of existing results on the topic. In Section~\ref{sec:quasi}, we introduce a geometric object in the projective plane, which we call a \textit{balanced quasi-arc}. Intuitively, a balanced quasi-arc is a set of projective points in which three specific points, referred to as \textit{fundamental points}, are in the span of many subsets of other points in the set, more so than any of the \textit{non-fundamental points}. This characteristic of the set of points increases the likelihood of recovering these fundamental points when sampling uniformly at random from the set, ultimately leading to codes that perform well regarding the Random Access Problem. To the best of our knowledge, the codes we obtain based on this construction outperform any of the existing constructions of codes in terms of the Random Access Problem. All of this will be discussed in detail in Section~\ref{sec:constr}. In Section~\ref{sec:conj}, we investigate the parameters of a class of codes with a rate of $1/2$, first constructed in~\cite{bar2023cover}, that have a random access expectation strictly smaller than their dimension $k$. By closely examining these parameters, we are able to resolve a conjecture proposed in~\cite{bar2023cover} that states the ratio of their random access expectation to $k$ is strictly smaller than $0.9456$ as $k$ tends to infinity. The approach we take is inspired by a result from~\cite{gruica2024reducing} and it gives a new way of showing that the random access expectation is smaller than $k$. Finally, in Section~\ref{sec:concl} we conclude the paper and propose some possible directions for future research.

\section{Preliminaries} \label{sec:prelim}

Throughout this paper, we will use the following notations.

\begin{notation}
\begin{itemize}
    \item[(i)] $k$ and $n$ are positive integers with $1 \le k \le n$;
    \item[(ii)] $q$ denotes a prime power;
    \item[(iii)] $\F_q$ is the finite field with $q$ elements;
    \item[(iv)] $[n]$ denotes the set of integers from $1$ to $n$, i.e., $[n]:=\{1, \ldots,n\}$;
    \item[(v)] $H_n$ denotes the $n$-th harmonic number, i.e., $H_n=1+1/2+1/3+\dots + 1/n$.
\end{itemize}
\end{notation}

We study the expected sample size for uniformly random access queries in DNA storage systems. In DNA-based storage systems, the data is stored as a length-$k$ vector of sequences (called \textit{strands}) of length $\ell$ over the alphabet $\Sigma=\{A,C,G,T\}$, i.e., as a vector in $(\Sigma^{\ell})^k$. To make it more general, we consider our information strands to be elements of a finite field $\F_q$ (if we start with $\Sigma^\ell$ then one can think of an embedding into $\F_q$, i.e., then we need that $4^\ell$ divides $q$) and use a $k$-dimensional linear block
code $\mC \subseteq \F_q^n$
to encode an information vector
$ (x_1,\dots,x_k) \in \F_q^k$ to an encoded vector $(y_1,\dots,y_n) \in \F_q^n$.

Whenever the user needs to retrieve the stored information, the strands are amplified and then sequenced using DNA sequencing technology. This generates multiple copies for different strands, possibly with errors, referred to as \emph{reads}. In this paper, similarly as in~\cite{gruica2024reducing,bar2023cover}, we assume that no errors are introduced in neither the synthesizing nor the sequencing of the strands. The output of the reading process is a multiset of these reads, without any specific order. Given that DNA sequencing costs and throughput still trail behind other archival storage solutions, reducing the coverage depth for information recovery is of great interest.

The problem we focus on in this paper was first introduced in~\cite{bar2023cover} and was then studied in more detail in~\cite{gruica2024reducing}. In this problem, the goal is to retrieve a single information strand $x_i$ for $i \in [k]$. It has been demonstrated \cite{bar2023cover} that the expected sample size of a random access query in a DNA storage system can be reduced using an error-correcting code. Note that a related concept is investigated in~\cite{chandak2019improved}, where the authors examine the trade-offs between reading costs, tied to coverage depth, and writing costs. The non-random access version of the DNA coverage depth problem was later extended in~\cite{cohen2024optimizing} where composite DNA letters were considered~\cite{anavy2019data}, and in~\cite{preuss2024sequencing, sokolovskii2024coding} to handle combinatorial composite DNA shortmers~\cite{preuss2021efficient}. Another extension focusing on the random access setup was studied in~\cite{AGY24}; however, the objective there was to decode a collection of strands that together form a single file, rather than just a single strand.

Whenever the $k$ information strands are encoded using a generator matrix $G \in \F_q^{k \times n}$ one can think of an encoded strand as its corresponding column in the matrix~$G$ and recovering the~{$i$-th}~information strand corresponds to recovering the $i$-th standard basis vector, that is, it should belong to the span of the already recovered columns of $G$. This justifies that when we care about recovering the $i$-th information strand, we only care about the set of columns of the matrix $G$, and not the order in which they appear. Moreover, the set of columns can be seen as a set of vectors in $\F_q^k$ and hence \textit{points} in $\mathrm{PG}(k-1, q)$, the $(k-1)$-dimensional projective space over $\mathbb{F}_q$, as we care only about their span. We recall the definition of $\mathrm{PG}(k-1, q)$ and of other related objects that we use throughout this paper.

\begin{definition}
The projective geometry $\mathrm{PG}(k-1, q)$ is the projective space of dimension~$k-1$ over a finite field $\mathbb{F}_q$ with $q$ elements. It is defined as the set of equivalence classes of the non-zero vectors in $\mathbb{F}_q^k$, where two vectors are considered equivalent if they differ by a scalar multiple.
\begin{itemize}
    \item[(i)] A \textbf{point} in 
$\mathrm{PG}(k-1, q)$ is an equivalence class of non-zero vectors in $\mathbb{F}_q^k$ under the equivalence relation $x \sim \lambda x$ for all non-zero scalars $\lambda \in \mathbb{F}_q$.
    \item[(ii)]  A \textbf{line} in $\mathrm{PG}(k-1, q)$ is the set of points corresponding to all scalar multiples of the vectors in a 2-dimensional subspace of $\mathbb{F}_q^k$. If $P_1$ and $P_2$ generate a line $\ell$, we sometimes write $\ell=P_1P_2$. 
    \item[(iii)] We say that a point $P$ in $\mathrm{PG}(k-1, q)$ is \textbf{fundamental} if it is a non-zero multiple of a standard basis vector in $\F_q^k$, and we say that a point is \textbf{non-fundamental} if it is not. 
    \item[(iv)] We say that a multiset of points $\mG=\{P_1,\dots,P_n\} \subseteq \mathrm{PG}(k-1, q)$ has \textbf{rank $k$} if the points satisfy $\langle P_1,\dots,P_n\rangle =\mathrm{PG}(k-1, q)$.
\end{itemize}
\end{definition}

Now we can formally define the problem we study in this paper. 

\begin{problem}[{\textbf{The random access problem}}] \label{prob:1}
Let $\mG = \{P_1,\dots,P_n\} \subseteq \mathrm{PG}(k-1, q)$ be a set of $n$ points of rank $k$. Suppose that the elements of $\mG$ are drawn uniformly at random, meaning that each point has probability $1/n$ of being drawn and points can be drawn multiple times. For $P \in \mG$, let $\tau_P(\mG)$ denote the random variable that counts the minimum number of points of $\mG$ that are drawn until the point $P$ is in their span. Compute the expectation~$\E[\tau_P(\mG)]$ and find point sets that reduce the random access expectation for fundamental points. Also, compute what would be the worst-case-scenario for the expected number of draws needed to recover any of the fundamental points, $\max\{\E[\tau_P(\mG)] : P \in \mG, \textnormal{$P$ is fundamental}\}$.

\end{problem}

\begin{notation}
In the sequel, let $\mG = \{P_1, \dots, P_n\} \subseteq \mathrm{PG}(k-1, q)$ denote a set of $n$ points in $\mathrm{PG}(k-1, q)$ of rank $k$. The set $\mG$ is often a multiset of points; however, with a slight abuse of notation, we always write $\mG \subseteq \mathrm{PG}(k-1, q)$. As our primary focus is on retrieving information strands, or equivalently, fundamental points, we assume that all $k$ fundamental points are contained in $\mG$, unless otherwise stated. 
Furthermore, when the random access expectation is the same for all fundamental points, we denote their corresponding random variable by~$\tau_{F}(\mG)$. Similarly, if the random access expectation is the same for all non-fundamental points, we denote their random variable by~$\tau_{N}(\mG)$.

\end{notation}

\begin{remark} \label{rem:points}
Differently to previous approaches -- such as examining codes in~\cite{bar2023cover} or analyzing generator matrices in~\cite{gruica2024reducing} -- we focus solely on the points formed from the columns of the generator matrices of codes, and so our problem is stated entirely in terms of subsets of $\mathrm{PG}(k-1, q)$ of rank $k$.
\end{remark}

We start by surveying the results from~\cite{bar2023cover} contributing to the random access problem. Among other things, it was shown that by connecting to the coupon collector's problem~\cite{erdHos1961classical, felleb1968introduction, flajolet1992birthday, newman1960double}, if the $k$ information strands are encoded by an MDS code, then the expected number of reads to decode all $k$ information strands is $n(H_n-H_{n-k})$, where $H_i$ is the $i$-th harmonic number.
This result is optimal for minimizing the expected number of reads. However, analyzing MDS codes for the random access problem, it was proven that the expected number of reads is~$k$, which is also what an identity code would give. Moreover, in~\cite{bar2023cover} it was shown that codes with expectation $c  k$, for $c < 1$, exist. In Section~\ref{sec:conj}, we examine in more detail the construction of a class of codes with rate $1/2$ that has a low random access expectation, thereby proving an open conjecture proposed by the authors of~\cite{bar2023cover}.

As a step towards a better understanding of the properties that give a \textit{good} code in terms of having low random access expectation, in~\cite{gruica2024reducing} a general formula for the expectation was given. Let $\mG=\{P_1,\dots,P_n\}\subseteq \mathrm{PG}(k-1, q)$ be a multiset of points of rank $k$. For~$P \in \mG$ and $1 \le s \le n-1$, let 
\begin{align*}
\alpha_{P}(\mG,s) := {|\{S \subseteq [n] : |S| = s, \, P \in \langle P_j : j \in S\rangle \}|}.
\end{align*}
The formula we will use throughout this paper for computing $\E[\tau_P(\mG)]$ uses the values~$\alpha_{P}(\mG,s)$ we just introduced. Note that in~\cite{gruica2024reducing} the formula was stated slightly differently. We restate it in the language of this paper.

\begin{lemma}[\textnormal{\cite[Lemma 1]{gruica2024reducing}}] \label{lem:fi}
For a multiset of points $\mG=\{P_1,\dots,P_n\}\subseteq \mathrm{PG}(k-1, q)$ of rank $k$ and for all $P \in \mG$ we have
\begin{align*}
    \E[\tau_P(\mG)] =  n H_n - \sum_{s=1}^{n-1} \displaystyle \frac{\alpha_{P}(\mG,s)}{\binom{n-1}{s}}.
\end{align*}
\end{lemma}

In \cite{gruica2024reducing}, it was shown that, on average, the random access expectation is always $k$ (see Theorem~\ref{thm:sumra}). If all the points in our set are equally likely to be recovered, we call the code arising from this point set a \textit{recovery balanced code}. In~\cite{gruica2024reducing} it was shown that MDS codes, the simplex code, the Hamming code, the binary Reed-Muller code, and the binary Golay code are all recovery balanced, meaning that for these codes the random access expectation is $k$. From this observation it became evident that in order to construct point sets that perform well in the random access setting, one needs to create an imbalance within the set. Specifically, we want a point set in which the fundamental points are more likely to be recovered than the non-fundamental points.

\section{Balanced Quasi-Arcs in Projective Planes} \label{sec:quasi}

In this section we give a construction of point sets in the \textit{projective plane} $\mathrm{PG}(2, q)$ that, if one considers the matrix whose column set is this point set, will give generator matrices of codes that perform well in terms of the random access problem, as we will see in the next sections. The main idea behind our construction is that we want that the three fundamental points are contained in more lines arising from our set of points, than the non-fundamental points. This will lead to an imbalance in those codes, making it more likely to recover fundamental points, and thus leading to a random access expectation that is strictly smaller than $k$.

\begin{definition}[Balanced quasi-arc]
    Let $\mG \subseteq \mathrm{PG}(2,q)$ be a set of points, let $\mF:=\{E_1, E_2, E_3\}$ be a set of non-collinear points in $\mathrm{PG}(2,q)$ and let $x \ge 0$ be an integer. We call $\mG$ a \textbf{balanced quasi-arc of weight $x$} if there exists a partition of $\mG$ of the form $\mG=\mF \cup \mG_1 \cup \mG_2 \cup \mG_3$ such that
    \begin{itemize}
        \item[(i)] $|\mG_1|=|\mG_2|=|\mG_3|=x$; 
        \item[(ii)] $\mG_1 \subseteq E_1E_2$, $\mG_2 \subseteq E_2E_3$, $\mG_3 \subseteq E_3E_1$;
        \item[(iii)] for any line $\ell$ in $\mathrm{PG}(2,q)$ different from $E_iE_j$ for any $i,j$, we have $|\mathcal{G}\cap \ell|\leq 2$.
    \end{itemize}
    Note the condition in (iii) is equivalent to the property that for $P_1 \in \mG_1$, $P_2 \in \mG_2$, $P_3 \in \mG_3$, we always have $\langle P_1,P_2,P_3 \rangle = \mathrm{PG}(2,q)$.
\end{definition}

In the sequel, we will, without loss of generality, choose $E_1,E_2$ and $E_3$ to be the fundamental points of $\mathcal{G}$ and then we call the lines $E_iE_j$, for any $i\ne j$, \textbf{fundamental lines}.

Roughly speaking, a balanced quasi-arc of weight $x$ behaves like an \textit{arc} of the projective plane with respect to all the lines apart from the fundamental lines, which are $(x+2)$-secant lines (i.e. lines that intersects $\mG$ in $x+2$ points).
A balanced quasi-arc of weight $0$ corresponds to a set of three non-collinear points in~$\mathrm{PG}(2,q)$.

\begin{remark}
    Observe that a balanced quasi-arc of weight $x$ in $\mathrm{PG}(2,q)$ has size $3x+3$ and the notion of fundamental points is well-defined as $E_1,E_2$ and $E_3$ are the only points of $\mathcal{G}$ lying on two $(x+2)$-secant lines.
\end{remark}

\begin{proposition}\label{prop:boundbalancedarc}
    Let $\mG$ be a balanced quasi-arc of weight $x>0$ in $\mathrm{PG}(2,q)$. We have that 
    \[x \le \left\lfloor\frac{q-1}2\right\rfloor.\]
\end{proposition}
\begin{proof}
    Denote by $E_1,E_2$ and $E_3$ the fundamental points of $\mG$.
    Consider any point $P \in \mathcal{G}\setminus \mathcal \{E_1,E_2,E_3\}$. Note that through $P$ there is a $(x+2)$-secant line and all the other lines share with $\mathcal{G}$ at most one other point. Taking into account that there are $q+1$ lines through~$P$, this implies that the size of $\mathcal{G}$ is bounded as follows.
    \[ |\mathcal{G}|\leq x+1 + q +1. \]
    Since $|\mathcal{G}|=3x+3$ we get $x\leq (q-1)/2$.
\end{proof}

We will show that the above bound is tight for all the possible values of $q$. The construction we are going to consider is the following.

\begin{definition}[$(S_1,S_2,S_3)$-projective triangle] \label{def:tri}
    Let $S_1,S_2,S_3$ be three subsets of $\F_q$. We define the $(S_1,S_2,S_3)$-\textbf{projective triangle} as the set of the following points:
    \begin{itemize}
    \item $E_1=(1:0:0)$, $E_2=(0:1:0)$ and $E_3=(0:0:1)$;
    \item $(0:1:-s_1)$, where $s_1 \in S_1$;
    \item $(-s_2:0:1)$, where $s_2 \in S_2$;
    \item $(1:-s_3:0)$, where $s_3 \in S_3$.
    \end{itemize} 
\end{definition}

\begin{remark}
    The $(S_1,S_2,S_3)$-projective triangle from Definition~\ref{def:tri} is inspired by the projective triangle which, in our terminology, is a $(S,S,N)$-projective triangle in PG$(2,q)$ where $q$ is odd, $S$ is the set of squares in $\mathbb{F}_q$, and $N$ is the set of non-squares in $\mathbb{F}_q$. The projective triangle turns out to be a \textit{minimal blocking set} in $\mathrm{PG}(2,q)$; A minimal blocking set in $\mathrm{PG}(2,q)$ is a point set, where every line of $\mathrm{PG}(2,q)$ meets this set in at least one point and the set does not contain any line, and which is minimal with respect to this property and set-theoretic inclusion. When $q$ is an odd prime, Blokhuis in \cite{blokhuis1994size} proved that the size of a minimal blocking set is at least $3(q-1)/2+3$ and so the projective triangle attains the lower bound with equality; see also \cite{hirschfeld1998projective}.
\end{remark}

\begin{proposition}
    Let $H$ be a subgroup of $(\F_q^*,\cdot)$ and suppose there exists $\tilde{H}\subset \F_q^*\setminus H$ such that $|H|=|\tilde{H}|=x$. Then the $(H,H,\tilde{H})$-projective triangle is a balanced quasi-arc of weight~$x$.
\end{proposition}
\begin{proof}
    Note that the points of the $(H,H,\tilde{H})$-projective triangle $\mG$ can be written as
    \[\mG=\mF \cup \mG_1 \cup \mG_2 \cup \mG_3,\]
    where $\mF:=\{E_1, E_2, E_3\}$, $\mathcal{G}_1=\{(0:1:-s_1)\colon s_1\in H\}$, $\mathcal{G}_2=\{(-s_2:0:1)\colon s_2\in H\}$ and $\mathcal{G}_3=\{(1:-s_3:0)\colon s_3\in \tilde{H}\}$.
    Indeed, $P_1,P_2,P_3$ are collinear if and only if
    $$\det\begin{pmatrix}
        0 & 1 & -s_1\\
        -s_2 & 0 & 1\\
        1 & -s_3 & 0
    \end{pmatrix}=1-s_1s_2s_3=0,$$
    implying that $s_1s_2s_3=1$. This is not possible as $H$ is a subgroup of $(\F_q^*,\cdot)$, $s_1s_2 \in H$ and~$s_3 \notin H$.
\end{proof}

We can now show that we can choose $H$ in such a way that we get a construction of balanced quasi-arcs satisfying the equality in Proposition \ref{prop:boundbalancedarc}.

\begin{corollary}\label{cor:construction}
Let $q$ be an odd prime power and consider $H$ as the subgroup of squares in $(\F_q^*,\cdot)$ and $\tilde{H}=\F_q^*\setminus H$.
We have that the $(H,H,\tilde{H})$-projective triangle is a balanced quasi-arc of weight $\frac{q-1}2$.
\end{corollary}

\begin{example} \label{ex:quasiarcs}
We give examples for balanced quasi-arcs for $q$ odd and $q$ even based on the construction in Corollary~\ref{cor:construction}. 
\begin{itemize}
    \item[(i)] Let $q=5$. Note that $1$ and $4$ are the only squares in $(\F_5^*, \cdot)$, and so we let $H:=\{1,4\}$ and $H':=\{2,3\}$. Then by Corollary~\ref{cor:construction} the $(H,H,\tilde{H})$-projective triangle is a balanced quasi-arc of weight $(q-1)/2=2$. More precisely, $\mG=\mF \cup \mG_1 \cup \mG_2 \cup \mG_3$, where 
    $\mF$ is the set of fundamental points, and
    \begin{align*}
        \mG_1:=\{(0:1:4),(0:1:1)\}, \mG_2:=\{(4:0:1),(1:0:1)\}, \mG_3:=\{(1:3:0),(1:2:0)\},
    \end{align*}
    is a balanced quasi-arc.  
    \item[(ii)] Let $q=4$ and consider the finite field $\F_4$ made of the elements in $\{0,1,\alpha,\alpha+1\}$ where $\alpha$ is a root of the irreducible polynomial $x^2+x+1$ over $\F_2$. Then we let $H:=\{1\}$, $\tilde{H}:=\{\alpha\}$ and we obtain that $\mG=\mF \cup \mG_1 \cup \mG_2 \cup \mG_3$, where 
    $\mF$ is the set of fundamental points, and
    \begin{align*}
        \mG_1:=\{(0:1:1)\}, \mG_2:=\{(1:0:1)\}, \mG_3:=\{(1:\alpha:0)\},
    \end{align*}
    is a balanced quasi-arc.
\end{itemize}
\end{example}

\begin{remark}
    Note that if $q=2$ then $x=0$ and thus the bound of Proposition \ref{prop:boundbalancedarc} is trivially tight in this case.
\end{remark}

As a corollary we also get the existence of balanced quasi-arcs in $\mathrm{PG}(2,q)$ of weight $x\leq \lfloor \frac{q-1}2\rfloor$, for any $x$.

\begin{corollary} \label{cor:q}
    For any $q$ odd prime power and $x\leq  \frac{q-1}2$ a positive integer we have that there exists a balanced quasi-arc in $\mathrm{PG}(2,q)$ of weight $x$.
\end{corollary}
\begin{proof}
    By Corollary \ref{cor:construction}, we know that there exists a balanced quasi-arc $\mG$ in $\mathrm{PG}(2,q)$ of weight $ \frac{q-1}2$. 
    By definition, $\mG$ can be partitioned into four sets
    \[ \mG=\mF \cup \mG_1 \cup \mG_2 \cup \mG_3, \]
    where $\mF$ is the set of fundamental points and $\mG_1, \mG_2, \mG_3$ are sets of points of size $\lfloor \frac{q-1}2\rfloor$.
    Consider $\mG_i'$ a subset of $\mG_i$ of size $x$, for any $i \in \{1,2,3\}$.
    It is easy to see that the set $\mathcal{G}'=\mF \cup \mG_1' \cup \mG_2' \cup \mG_3'$ is a balanced quasi-arc in $\mathrm{PG}(2,q)$ of weight $x$.
\end{proof}

\begin{remark}
    For the case $q$ even, a possible choice for $H$ is the largest subgroup of the multiplicative group of $(\fq^*,\cdot)$ and as $\tilde{H}$ a subset of the same size in $\fq^*\setminus H$.
    This clearly depends on the prime factorization of $q-1$; see \cite{szHonyi1991combinatorial,cameron2020four}.
\end{remark}

\section{A Construction for $k=3$}  \label{sec:constr}

In this section, we show that if we consider the set of points $\mG$ as a balanced quasi-arc, we can compute $\alpha_F(\mG,s)$ for all $1 \le s\le n-1$. Besides obtaining (relatively) clean formulas for the values of $\alpha_F(\mG,s)$, we demonstrate that the random access expectation for these sets of points is strictly less than $k$. We separately address the case where the fundamental points appear only once in our sets, and the case where multiplicities are added to the fundamental points. We will also need the following notation.

\begin{notation}
From now on, both in this section and in Section~\ref{sec:conj}, the point sets that we analyze are \emph{balanced} within the set of fundamental points, and within the set of non-fundamental points. Formally, what we mean with this is that for a point set $\mG$ of rank $k$ and of cardinality~$n$, and for any fundamental points $E,E'$, we have
\begin{align*}
    \alpha_{E}(\mG,s)=\alpha_{E'}(\mG,s) \quad \textnormal{for all $1 \le s \le n-1$},
\end{align*}
and for any non-fundamental points $P,Q \in \mG$, we have
\begin{align*}
    \alpha_{P}(\mG,s)=\alpha_{Q}(\mG,s) \quad \textnormal{for all $1 \le s \le n-1$}.
\end{align*}
This implies that, when computing the random access expectation of any of the points in the set $\mG$, the expectation for all the fundamental points, and the expectation for all the non-fundamental points, is the same, respectively. Therefore, in the sequel, we often write $\alpha_{F}(\mG,s)$ when the point we wish to recover is fundamental, and we write $\alpha_{N}(\mG,s)$ when it is non-fundamental.
\end{notation}

\subsection{Distinct points}\label{subsec:dp}

In this subsection, we let $x \ge 0$ be an integer and we consider a balanced quasi-arc of weight~$x$ which we denote by $\mG_x=\mF \cup \mG_1 \cup \mG_2 \cup \mG_3$. Note that by Corollary~\ref{cor:q} we need $\left\lfloor \frac{q-1}{2}\right\rfloor \ge x$ to be satisfied in order for a quasi-arc of weight $x$ to exist over $\F_q$.

\begin{proposition} \label{prop:alphasimple}
    We have 
    \begin{align*}
        \alpha_F(\mG_x,1) &= 1, \\
        \alpha_F(\mG_x,2) &= 2\binom{x+2}{2}+x,  \\
        \alpha_F(\mG_x,s) &= \binom{3x+3}{s}-\binom{x+2}{s} \quad \textnormal{ if $3 \le s \le x+2$}, \\
        \alpha_F(\mG_x,s) &= \binom{3x+3}{s} \quad \textnormal{ if $x+2< s \le 3x+2$.}
    \end{align*}
\end{proposition}
\begin{proof}
Suppose that we want to recover the fundamental point $E_1$.

It is easy to see that $\alpha_{E_1}(\mG_x,1)=1$. 

For $s=2$, we can count the number of 2-sets in $\mG$ that recover $E_1$ by distinguishing between two cases: If the two points span $E_1E_2$ or $E_3E_1$, they will for sure recover $E_1$, and there are $\binom{x+2}{2}$ 2-sets for either of the lines. If this is not the case, the only way to recover $E_1$ with two points is when one of them is $E_1$. There are a total of $x$ such points that will not span either of the lines $E_1E_2$ or $E_3E_1$, i.e., all the points in $\mG_2$.

For $3 \le s \le x+2$, the only $s$-sets that do \emph{not} recover $E_1$ are those that only consist of points in $E_2E_3$. This gives that there are $\binom{3x+3}{s}-\binom{x+2}{s}$ sets that recover $E_1$. 

Finally, if $s > x+2$, it is not possible to only have points from $E_2E_3$ in the $s$-set, and so necessarily any $s$-set will recover $E_1$.

All of the computations in this proof can be done analogously for the fundamental points $E_2$ and $E_3$, and so this code is balanced within the set of fundamental points. Therefore we can write more generally $\alpha_F(\mG_x,s)$ instead of $\alpha_{E_1}(\mG_x,s)$ and obtain the statement of the proposition.
\end{proof}

Combining Proposition~\ref{prop:alphasimple} with Lemma~\ref{lem:fi} we obtain a closed formula for the random access expectation of the point set $\mG_x$ that only depends on the weight $x$ of the balanced quasi-arc $\mG_x$.
\begin{corollary} \label{cor:expconstr}
We have
\begin{align*}
        \E[\tau_F(\mG_x)] = 3+\frac{2}{3x+1}-\frac{2((x+2)(x+1)+x)}{(3x+2)(3x+1)}+\sum_{s=3}^{x+2}\prod_{i=0}^{s-1}\frac{x+2-i}{3x+2-i}.
\end{align*}
\end{corollary}
\begin{proof}
    We simplify the formula given in Lemma~\ref{lem:fi} for the values of $\alpha_F(\mG_x,s)$ computed in Proposition~\ref{prop:alphasimple}. We have
    \begin{align*}
        &\E[\tau_F(\mG_x)] = (3x+3)H_{3x+3}-\sum_{s=1}^{3x+2}\frac{\alpha_F(\mG_x,s)}{\binom{3x+2}{s}} \\
        &= \sum_{s=1}^{3x+3} \frac{3x+3}{s}-\frac{1}{3x+2}-\frac{2((x+2)(x+1)+x)}{(3x+2)(3x+1)}-\sum_{s=3}^{x+2}\frac{\binom{3x+3}{s}}{\binom{3x+2}{s}}+\sum_{s=3}^{x+2}\frac{\binom{x+2}{s}}{\binom{3x+2}{s}}-\sum_{s=x+3}^{3x+2}\frac{\binom{3x+3}{s}}{\binom{3x+2}{s}} \\
        &= \sum_{s=1}^{3x+3} \frac{3x+3}{s}-\frac{1}{3x+2}-\frac{2((x+2)(x+1)+x)}{(3x+2)(3x+1)}-\sum_{s=3}^{3x+2}\frac{\binom{3x+3}{s}}{\binom{3x+2}{s}}+\sum_{s=3}^{x+2}\frac{\binom{x+2}{s}}{\binom{3x+2}{s}}\\
        &= \sum_{s=1}^{3x+3} \frac{3x+3}{s}-\frac{1}{3x+2}-\frac{2((x+2)(x+1)+x)}{(3x+2)(3x+1)}-\sum_{s=3}^{3x+2}\frac{3x+3}{3x+3-s}+\sum_{s=3}^{x+2}\frac{\binom{x+2}{s}}{\binom{3x+2}{s}}
\end{align*}
where in the last equality we used the fact that for integers $a > b \ge 1$ we have $\binom{a}{b}/\binom{a-1}{b}=a/(a-b)$. We further simplify to obtain the the above expression is equal to
\begin{align*}
        \E[\tau_F(\mG_x)] &= \sum_{s=1}^{3x+3} \frac{3x+3}{s}-\frac{1}{3x+2}-\frac{2((x+2)(x+1)+x)}{(3x+2)(3x+1)}-\sum_{s=1}^{3x}\frac{3x+3}{s}+\sum_{s=3}^{x+2}\frac{\binom{x+2}{s}}{\binom{3x+2}{s}} \\
        &= \sum_{s=3x+1}^{3x+3} \frac{3x+3}{s}-\frac{1}{3x+2}-\frac{2((x+2)(x+1)+x)}{(3x+2)(3x+1)}+\sum_{s=3}^{x+2}\frac{\binom{x+2}{s}}{\binom{3x+2}{s}} \\
        &= 1+\frac{3x+3}{3x+2}+\frac{3x+3}{3x+1}-\frac{1}{3x+2}-\frac{2((x+2)(x+1)+x)}{(3x+2)(3x+1)}+\sum_{s=3}^{x+2}\frac{\binom{x+2}{s}}{\binom{3x+2}{s}} \\
        &= 2+\frac{3x+3}{3x+1}-\frac{2((x+2)(x+1)+x)}{(3x+2)(3x+1)}+\sum_{s=3}^{x+2}\frac{\binom{x+2}{s}}{\binom{3x+2}{s}}.
\end{align*}
Using the definition of the binomial coefficient we obtain that 
\begin{align*}
    \frac{\binom{x+2}{s}}{\binom{3x+2}{s}} = \prod_{i=0}^{s-1}\frac{x+2-i}{3x+2-i}
\end{align*}
and so we further simplify to obtain that 
\begin{align*}
        \E[\tau_F(\mG_x)] &= 2+\frac{3x+3}{3x+1}-\frac{2((x+2)(x+1)+x)}{(3x+2)(3x+1)}+\sum_{s=3}^{x+2}\prod_{i=0}^{s-1}\frac{x+2-i}{3x+2-i} \\
        &= 3+\frac{2}{3x+1}-\frac{2((x+2)(x+1)+x)}{(3x+2)(3x+1)}+\sum_{s=3}^{x+2}\prod_{i=0}^{s-1}\frac{x+2-i}{3x+2-i}. \qedhere
    \end{align*}
\end{proof}

We can use the formula in Corollary~\ref{cor:expconstr} in order to compute the random access expectation of the point sets of Example~\ref{ex:quasiarcs}.

\begin{example}
\begin{itemize}
    \item[(i)] Let $\mG$ be defined as the balanced quasi-arcs in Example~\ref{ex:quasiarcs} (i) of weight 2 over $\F_5$. We have $\E[\tau_F(\mG)] \approx 2.87$.
    \item[(ii)] Let $\mG$ be defined as the balanced quasi-arcs in Example~\ref{ex:quasiarcs} (ii) of weight 1 over $\F_4$. We have $\E[\tau_F(\mG)] = 2.9$.
\end{itemize}
\end{example}

The formula of Corollary~\ref{cor:expconstr} is not easy to evaluate explicitly for general $x$. However, for large $x$, we can get an asymptotic upper bound that reads as follows.

\begin{theorem}
    We have $\lim_{x\to \infty}{\E[\tau_F(\mG_x)]} \le 3-1/6 \approx 0.9\overline{44} k.$
\end{theorem}
\begin{proof}
We have
\begin{align*}
    \E[\tau_F(\mG)] &= 3+\frac{2}{3x+1}-\frac{2((x+2)(x+1)+x)}{(3x+2)(3x+1)}+\sum_{s=3}^{x+2}\prod_{i=0}^{s-1}\frac{x+2-i}{3x+2-i} \\
    &= 3+\frac{2}{3x+1}-\frac{2((x+2)(x+1)+x)}{(3x+2)(3x+1)}+\frac{(x+2)(x+1)}{(3x+2)(3x+1)}\sum_{s=3}^{x+2}\prod_{i=2}^{s-1}\frac{x+2-i}{3x+2-i}.
\end{align*}
Moreover,
\begin{align*}
        \frac{x+2-i}{3x+2-i} \le \frac{1}{3}
    \end{align*}
    for all $i \ge 2$. Because of this, and also because of Corollary~\ref{cor:expconstr}, we obtain
    \begin{align*}
         \E[\tau_F(\mG_x)] &\le 3+\frac{2}{3x+1}-\frac{2((x+2)(x+1)+x)}{(3x+2)(3x+1)}+\frac{(x+2)(x+1)}{(3x+2)(3x+1)}\sum_{s=3}^{x+2}\prod_{i=2}^{s-1}\frac{1}{3}\\
         &= 3+\frac{2}{3x+1}-\frac{2((x+2)(x+1)+x)}{(3x+2)(3x+1)}+\frac{(x+2)(x+1)}{(3x+2)(3x+1)}\sum_{s=3}^{x+2}\left(\frac{1}{3}\right)^{s-2} \\
         &= 3+\frac{2}{3x+1}-\frac{2((x+2)(x+1)+x)}{(3x+2)(3x+1)}+\frac{(x+2)(x+1)}{(3x+2)(3x+1)}\left(\sum_{s=0}^{x}\left(\frac{1}{3}\right)^{s}-1\right)  \\
         &= 3+\frac{2}{3x+1}-\frac{2((x+2)(x+1)+x)}{(3x+2)(3x+1)}+\frac{(x+2)(x+1)}{(3x+2)(3x+1)}\left(\frac{1-(1/3)^{x+1}}{1-1/3}-1\right) \\
         &= 3+\frac{2}{3x+1}-\frac{2((x+2)(x+1)+x)}{(3x+2)(3x+1)}+\frac{(x+2)(x+1)}{(3x+2)(3x+1)}\left(\frac{1}{2}-\frac{1}{2}\left(\frac{1}{3}\right)^x\right).
    \end{align*}
    As $x \to \infty$ we have
    \begin{align*}
         3+\frac{2}{3x+1}-\frac{2((x+2)(x+1)+x)}{(3x+2)(3x+1)}+\frac{(x+2)(x+1)}{(3x+2)(3x+1)}\left(\frac{1}{2}-\frac{1}{2}\left(\frac{1}{3}\right)^x\right) \sim 3-\frac{2}{9}+\frac{1}{18}= 3-\frac{1}{6}
    \end{align*}
    proving the statement of the theorem.
\end{proof}


\subsection{Fundamental points with multiplicity} \label{subsec:mp}

In this subsection, we improve the construction of the point set considered in Subsection~\ref{subsec:dp} by using a balanced quasi-arc of weight $x$ and adding an additional $y - 1$ copies of each fundamental point to the set where $x \ge 0$ and $y \ge 1$ are integers, resulting in a point set where each fundamental point has multiplicity $y$. We denote this point set by $\mG_{x,y} = \mF_y \cup \mG_1 \cup \mG_2 \cup \mG_3$, where $\mF_y$ contains $y$ copies of $E_1$, $E_2$, and $E_3$, and $\mG_1$, $\mG_2$, and $\mG_3$ satisfy the properties needed for $\mF \cup \mG_1 \cup \mG_2 \cup \mG_3$ to form a balanced quasi-arc. Adding multiplicities to the fundamental points creates a greater imbalance in this multiset, making it even more likely to recover fundamental points. We proceed here similarly to Subsection~\ref{subsec:dp} and start with giving closed formulas for $\alpha_F(\mG_{x,y},s)$ for all $1 \le s \le n-1$.

\begin{proposition} \label{prop:alphamult}
    We have 
    \begin{align*}
        \alpha_F(\mG_{x,y},1) &= y, \\
        \alpha_F(\mG_{x,y},2) &= 2\left( xy+\binom{x}{2} \right)+y(3x+2y)+y(y-1)/2,  \\
        \alpha_F(\mG_{x,y},s) &= \binom{3x+3y}{s}-\binom{x+2y}{s}-2\binom{y}{s-1}x \quad  \textnormal{ if $3 \le s \le x+2y$}, \\
        \alpha_F(\mG_{x,y},s) &= \binom{3x+3y}{s} \quad \textnormal{ if $x+2y < s \le 3x+3y-1$.}
    \end{align*}
\end{proposition}
\begin{proof}
As in the proof of Proposition~\ref{prop:alphasimple}, $\alpha_{E_1}(\mG_{x,y},s)$ for any $1 \le s \le n-1$ does not depend on $E_1$, and so w.l.o.g. suppose that we want to recover the fundamental point $E_1$.

It is easy to see that $\alpha_F(\mG_{x,y},1)=y$ since by definition every fundamental point shows up $y$ times in the set $\mG_{x,y}$.

As before, for $s=2$, we can count the number of 2-sets in $\mG$ that recover $E_1$ by distinguishing between three cases: If the two points span $E_1E_2$ or $E_3E_1$, they will for sure recover $E_1$, and there are $xy+\binom{x}{2}$ 2-sets for either of the lines without considering that $E_1$ could be one of the points. If $E_1$ is one of the two points, then any of the remaining $3x+2y$ points will give a recovering 2-set. Finally, if $y\ge 2$, the only 2-set that recovers $E_1$ that we did not count yet, is the 2-set made of only $E_1$. There are $y(y-1)/2$ such sets.

For $3 \le s \le x+2y$, the only $s$-sets that do \emph{not} recover $E_1$ are those that only consist of points in $E_2E_3$, and those that consist of $(s-1)$-sets of either $E_2$ or $E_3$, together with a point on the fundamental line that does not go through the $E_2$ or $E_3$, respectively. In total this gives $\binom{3x+3y}{s}-\binom{x+2y}{s}-2\binom{y}{s-1}x$ such $s$-sets.

Finally, if $s > x+2y$, it is not possible to only have points from $E_2E_3$ in the $s$-set, and so we conclude that any $s$-set will recover $E_1$.
\end{proof}

\begin{corollary} \label{cor:expconstrm}
We have
\begin{align*}
        \E[\tau_F(\mG_{x,y})] &= 3+\frac{2}{3x+3y-2} - \frac{y-1}{3x+3y-1} -\frac{2\left( xy+\binom{x}{2} \right)+y(3x+2y)+y(y-1)/2}{\binom{3x+3y-1}{2}} \\
        &\quad +\sum_{s=3}^{x+2y}\prod_{i=0}^{s-1}\frac{x+2y-i}{3x+3y-i-1}+\sum_{s=3}^{y+1}\frac{2\binom{y}{s-1}x}{\binom{3x+3y-1}{s}}
\end{align*}
\end{corollary}
\begin{proof}
    We simplify the formula given in Lemma~\ref{lem:fi} for the values of $\alpha(\mG_{x,y},s)$ computed in Proposition~\ref{prop:alphamult}. Similarly to the proof of Corollary~\ref{cor:expconstr}, we obtain
    \begin{align*}
        &\E[\tau_F(\mG_{x,y})] = (3x+3y)H_{3x+3y}-\sum_{s=1}^{3x+3y-1}\frac{\alpha(\mG_{x,y},s)}{\binom{3x+3y-1}{s}} \\
        &= \sum_{s=1}^{3x+3y} \frac{3x+3y}{s} - \frac{y}{3x+3y-1}-\frac{2\left( xy+\binom{x}{2} \right)+y(3x+2y)+y(y-1)/2}{\binom{3x+3y-1}{2}} \\
        &\quad - \sum_{s=1}^{3x+3y-3}\frac{3x+3y}{s}+\sum_{s=3}^{x+2y}\prod_{i=0}^{s-1}\frac{x+2y-i}{3x+3y-i-1} +\sum_{s=3}^{y+1}\frac{2\binom{y}{s-1}x}{\binom{3x+3y-1}{s}}\\
        &= 1+\frac{3x+3y}{3x+3y-1}+\frac{3x+3y}{3x+3y-2}- \frac{y}{3x+3y-1} \\
        &\quad -\frac{2\left( xy+\binom{x}{2} \right)+y(3x+2y)+y(y-1)/2}{\binom{3x+3y-1}{2}}+\sum_{s=3}^{x+2y}\prod_{i=0}^{s-1}\frac{x+2y-i}{3x+3y-i-1}+\sum_{s=3}^{y+1}\frac{2\binom{y}{s-1}x}{\binom{3x+3y-1}{s}} \\
        &= 3+\frac{2}{3x+3y-2} - \frac{y-1}{3x+3y-1} -\frac{2\left( xy+\binom{x}{2} \right)+y(3x+2y)+y(y-1)/2}{\binom{3x+3y-1}{2}} \\
        &\quad +\sum_{s=3}^{x+2y}\prod_{i=0}^{s-1}\frac{x+2y-i}{3x+3y-i-1}+\sum_{s=3}^{y+1}\frac{2\binom{y}{s-1}x}{\binom{3x+3y-1}{s}}
    \end{align*}
    which proves the corollary.
\end{proof}

As in the previous subsection, the formula in Corollary~\ref{cor:expconstrm} is not easy to evaluate explicitly, as it depends on both $x$ and $y$ and how they relate to each other. Note that in the previous section, we treated the case where $x$ is arbitrary and $y=1$, and we provided an asymptotic upper bound for $\E[\tau_F(\mG_x)]$ as $x \to \infty$. Here, we will consider the case where $x = y$ and compute an asymptotic upper bound as $x \to \infty$. As we will discuss later in Remark~\ref{rem:optimalratio}, experimental results suggest that $x = y$ is not the optimal case for reducing the random access expectation of the multiset of points $\mG_{x,y}$.

\begin{theorem}
    We have $\lim_{x\to \infty}{\E[\tau_F(\mG_{x,x})]} \le 3-{53}/{150} \approx 0.88\overline{22}k$.
\end{theorem}
\begin{proof}
For $x=y$ the expression in Corollary~\ref{cor:expconstrm} reads and simplifies as follows.
\begin{align*}
        \E[\tau_F(\mG_{x,x})] &= 3+\frac{2}{6x-2} - \frac{x-1}{6x-1} -\frac{2\left( x^2+\binom{x}{2} \right)+5x^2+x(x-1)/2}{\binom{6x-1}{2}} \\
        &\quad +\sum_{s=3}^{3x}\prod_{i=0}^{s-1}\frac{3x-i}{6x-i-1}+\sum_{s=3}^{x+1}\frac{2\binom{x}{s-1}x}{\binom{6x-1}{s}} \\
         &= 3+\frac{2}{6x-2} - \frac{x-1}{6x-1} -\frac{17x^2-3x}{(6x-1)(6x-2)} \\
         &\quad +\frac{3x\left(3x-1\right)}{\left(6x-1\right)\left(6x-2\right)}\sum_{s=3}^{3x}\prod_{i=2}^{s-1}\frac{3x-i}{6x-i-1}+\sum_{s=3}^{x+1}\frac{2\binom{x}{s-1}x}{\binom{6x-1}{s}} \\
\end{align*}
We have
\begin{align*}
    \frac{3x-i}{6x-i-1} \le \frac{1}{2}
\end{align*}
for all $i \ge 1$, and so we obtain
\begin{align*}
    &\E[\tau_F(\mG_{x,x})]\le 3+\frac{2}{6x-2} - \frac{x-1}{6x-1} -\frac{17x^2-3x}{(6x-1)(6x-2)} +\frac{3x}{2\left(6x-1\right)}\sum_{s=3}^{3x}\prod_{i=2}^{s-1}\frac{1}{2}+\sum_{s=3}^{x+1}\frac{2\binom{x}{s-1}x}{\binom{6x-1}{s}} \\
    &= 3+\frac{2}{6x-2}-\frac{x-1}{6x-1} -\frac{17x^2-3x}{(6x-1)(6x-2)} +\frac{3x}{2\left(6x-1\right)}\sum_{s=1}^{3x-2}\left(\frac{1}{2}\right)^s+\sum_{s=3}^{x+1}\frac{2\binom{x}{s-1}x}{\binom{6x-1}{s}} \\
    &= 3+\frac{2}{6x-2}-\frac{x-1}{6x-1} -\frac{17x^2-3x}{(6x-1)(6x-2)} +\frac{3x}{2\left(6x-1\right)}\left(\sum_{s=0}^{3x-2}\left(\frac{1}{2}\right)^s-1\right)+\sum_{s=3}^{x+1}\frac{2\binom{x}{s-1}x}{\binom{6x-1}{s}} \\
    &= 3+\frac{2}{6x-2}-\frac{x-1}{6x-1} -\frac{17x^2-3x}{(6x-1)(6x-2)} +\frac{3x}{2\left(6x-1\right)}\left(1-\left(\frac{1}{2}\right)^{3x}\right)+\sum_{s=3}^{x+1}\frac{2\binom{x}{s-1}x}{\binom{6x-1}{s}} \\
\end{align*}
In order to give an asymptotic estimate for $x \to \infty$, we need the following claim.
\begin{claim} \label{cl:maria}
    We have
    \begin{align*}
        \sum_{s=3}^{x+1}\frac{2\binom{x}{s-1}x}{\binom{6x-1}{s}} = \frac{2x(x-1)(x+1)x(16x-2)}{(x+1)(6x-1)(6x-2)5x(5x-1)}.
    \end{align*}
\end{claim}
\begin{clproof}
Recall that $\binom{x}{s-1}=\frac{s}{x+1}\binom{x+1}{s}$; see e.g. \cite[Table 174, Equation 3]{concretemath}. Using this we get $$2x \sum_{s=3}^{x+1}\frac{\binom{x}{s-1}}{\binom{6x-1}{s}}=\frac{2x}{x+1} \sum_{s=3}^{x+1}s\frac{\binom{x+1}{s}}{\binom{6x-1}{s}}.$$
From \cite[page 173]{concretemath}, we then see that $\frac{\binom{x+1}{s}}{\binom{6x-1}{s}}=\frac{\binom{6x-1-s}{x+1-s}}{\binom{6x-1}{x+1}}$ and hence
\begin{align*}
    \frac{2x}{x+1} \sum_{s=3}^{x+1}s\frac{\binom{x+1}{s}}{\binom{6x-1}{s}}&=\frac{2x}{x+1}\sum_{s=3}^{x+1}s\frac{\binom{6x-1-s}{x+1-s}}{\binom{6x-1}{x+1}}=\frac{2x}{x+1} \binom{6x-1}{x+1}^{-1} \sum_{s=3}^{x+1}s \binom{6x-1-s}{x+1-s} \\
    &=\frac{2x}{x+1}\binom{6x-1}{x+1}^{-1} \sum_{T=0}^{x-2}(x+1-T) \binom{5x-2+T}{T} \\
    &=\frac{2x}{x+1} \binom{6x-1}{x+1}^{-1} \bigg((x+1)\sum_{T=0}^{x-2}\binom{5x-2+T}{T}-\sum_{T=0}^{x-2}T \binom{5x-2+T}{T}\bigg).
\end{align*}
From \cite[Table 174, Equation 8]{concretemath} we have that $(x+1) \sum_{T=0}^{x-2} \binom{5x-2+T}{T}=(x+1) \binom{6x-3}{x-2}=(x+1)(x-1)\frac{\binom{6x-3}{x-1}}{5x-1}$, and hence
\begin{align*}
    &\frac{2x}{x+1} \binom{6x-1}{x+1}^{-1} \bigg((x+1)\sum_{T=0}^{x-2}\binom{5x-2+T}{T}-\sum_{T=0}^{x-2}T \binom{5x-2+T}{T}\bigg) \\
    &=\frac{2x}{x+1}  \binom{6x-1}{x+1}^{-1} \bigg(\frac{(x+1)(x-1)}{5x-1}\binom{6x-3}{x-1}-\sum_{T=0}^{x-2}T \binom{5x-2+T}{T}\bigg).
\end{align*}
Finally using that $$\sum_{T=0}^{x-2}T \binom{5x-2+T}{T}=\frac{(x-2)(x-1)}{5x}\binom{6x-3}{x-1}$$ we get
\begin{align*}
    \frac{2x}{x+1} &\binom{6x-1}{x+1}^{-1} \bigg(\frac{(x+1)(x-1)}{5x-1}\binom{6x-3}{x-1}-\frac{(x-2)(x-1)}{5x}\binom{6x-3}{x-1}\bigg) \\
    &=\frac{2x(x-1)}{x+1}\binom{6x-1}{x+1}^{-1} \binom{6x-3}{x-1} \bigg(\frac{x+1}{5x-1} -\frac{x-2}{5x}\bigg) \\
    &= \frac{2x(x-1)}{x+1} \cdot \frac{(x+1)!(5x-2)!}{(6x-1)!} \cdot \frac{(6x-3)!}{(x-1)!(5x-2)!} \cdot \frac{5x^2+5x-5x^2+10x+x-2}{5x(5x-1)} \\
    &= \frac{2x(x-1)}{x+1} \cdot \frac{(x+1)x}{(6x-1)(6x-2)} \cdot \frac{16x-2}{5x(5x-1)}
\end{align*}
which proves the claim.
\end{clproof}
From Claim~\ref{cl:maria} we obtain that
\begin{align*}
    \sum_{s=3}^{x+1}\frac{2\binom{x}{s-1}x}{\binom{6x-1}{s}} \sim \frac{2 \cdot 16}{6^2 \cdot 5^2}=\frac{8}{225} \quad \textnormal{as $x \to \infty$.}
\end{align*}
From the latter asymptotic estimate, we finally have
\begin{align*}
    3+\frac{2}{6x-2}-\frac{x-1}{6x-1} -\frac{17x^2-3x}{(6x-1)(6x-2)} +\frac{3x}{2\left(6x-1\right)}\left(1-\left(\frac{1}{2}\right)^{3x}\right)+\sum_{s=3}^{x+1}\frac{2\binom{x}{s-1}x}{\binom{6x-1}{s}} \\
    \sim 3-1/6-17/36+1/4+8/225  = 3-53/150,
\end{align*}
proving the asymptotic upper bound of the theorem.
\end{proof}

Note that, to the best of our knowledge, the point set $\mG_{x,y}$ outperforms all previously known constructions of codes in terms of achieving a low random access expectation. The best-known construction prior to this was given in~\cite[Theorem 12]{bar2023cover}, also for 3-dimensional codes, achieving an expectation of approximately $0.89k$.

\begin{remark} \label{rem:optimalratio}
In this paper we investigate the asymptotic estimate of the formula in Corollary~\ref{cor:expconstrm} for the case where $x=y$. As a follow-up to this work, in~\cite{boruchovsky2025making} the authors provide an asymptotic estimate for ${\E[\tau_F(\mG_{x,y})]}$ following the same steps as we took, for the case $x \ne y$, showing that when $y=0.834x$ one obtains $\lim_{x \to \infty}{\E[\tau_F(\mG_{x,y})]} \le  0.881\overline{66} \cdot 3$; see~\cite[Corollary 2]{boruchovsky2025making}. We include a plot in Figure~\ref{fig:exp} that illustrates the asymptotics for $\E[\tau_F(\mG_{x,y})]$ for various values of $x$ and different ratios~$y/x$. 
\end{remark}

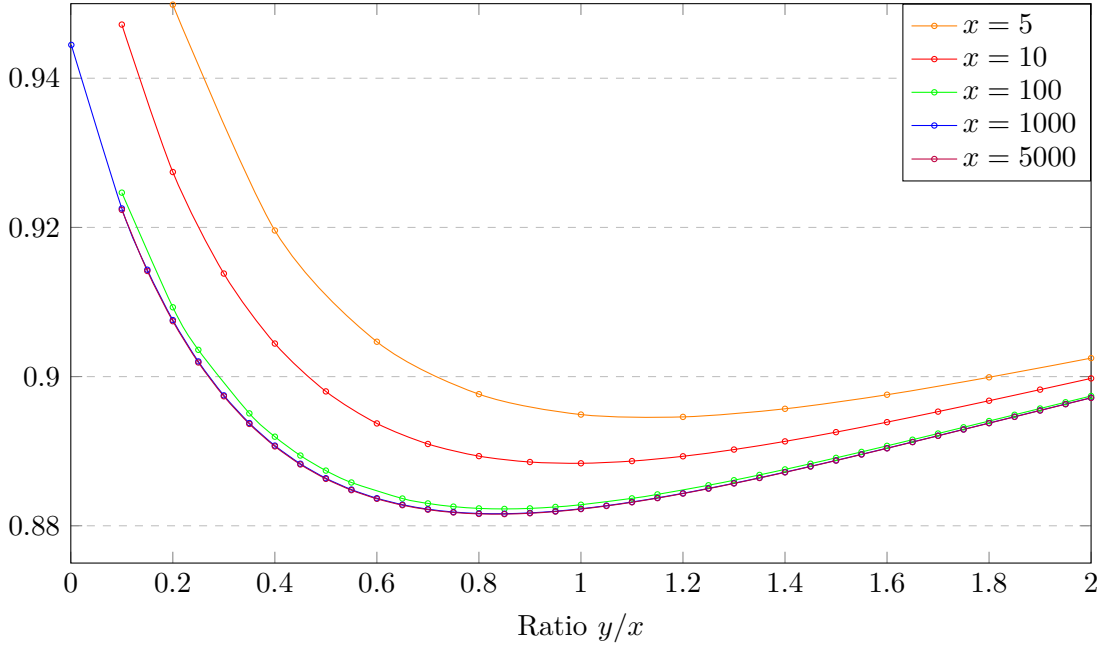
\begin{figure}[ht!]
\centering
\begin{tikzpicture}[scale=1]
\begin{axis}[legend style={at={(1,1)}, legend style={cells={align=left}}, anchor = north east, /tikz/column 2/.style={
                column sep=5pt}},
		legend cell align={left},
		width=15cm,height=9cm,
    xlabel={Ratio $y/x$},
    xmin=0, xmax=2,
    ymin=0.875, ymax=0.95,
    xtick={0,0.2,0.4,0.6,0.8,1,1.2,1.4,1.6,1.8,2},
    ytick={0.88,0.90,0.92,0.94,0.96},
    ymajorgrids=true,
    grid style=dashed,
    every axis plot/.append style={},  yticklabel style={/pgf/number format/fixed}
]
\addplot+[color=orange,mark=o,mark size=1pt,smooth]
coordinates {
(0.200000000000000000000000000000,0.949866310160427807486631016043)
(0.400000000000000000000000000000,0.919590643274853801169590643275)
(0.600000000000000000000000000000,0.904675035109817718513370687284)
(0.800000000000000000000000000000,0.897645832428441124093298006342)
(1.00000000000000000000000000000,0.894909688013136288998357963875)
(1.20000000000000000000000000000,0.894598382501608308059920963146)
(1.40000000000000000000000000000,0.895682410895392639814546508258)
(1.60000000000000000000000000000,0.897569046125922187042730336449)
(1.80000000000000000000000000000,0.899905667716964250982222740888)
(2.00000000000000000000000000000,0.902477481478849467905555456855)
};
\addplot+[color=red,mark=o,mark size=1pt,smooth]
coordinates {
(0.100000000000000000000000000000,0.947196620583717357910906298003)
(0.200000000000000000000000000000,0.927425515660809778456837280367)
(0.300000000000000000000000000000,0.913814108322117475435553238757)
(0.400000000000000000000000000000,0.904431067781517075483699360464)
(0.500000000000000000000000000000,0.898015937550821271751504309644)
(0.600000000000000000000000000000,0.893721983137835668238863112625)
(0.700000000000000000000000000000,0.890966299604173358326182246448)
(0.800000000000000000000000000000,0.889339126828872522063662342579)
(0.900000000000000000000000000000,0.888547195040225120929449616170)
(1.00000000000000000000000000000,0.888377326383504892313442033739)
(1.10000000000000000000000000000,0.888672468539952286491621421786)
(1.20000000000000000000000000000,0.889315567381605117454174057948)
(1.30000000000000000000000000000,0.890218502597782668019718063617)
(1.40000000000000000000000000000,0.891314364246759515135007761956)
(1.50000000000000000000000000000,0.892551976833379201616954112891)
(1.60000000000000000000000000000,0.893891961928923714667821099951)
(1.70000000000000000000000000000,0.895303870778554322858120326475)
(1.80000000000000000000000000000,0.896764071826323400012250407413)
(1.90000000000000000000000000000,0.898254177916060023988235189411)
(2.00000000000000000000000000000,0.899759863975789720782392985617)
    };
\addplot+[color=green,mark=o,mark size=1pt,smooth]
coordinates {
(0.100000000000000000000000000000,0.924662965055392090447474016717)
(0.200000000000000000000000000000,0.909294809352309276101483442153)
(0.250000000000000000000000000000,0.903589069518354035497888053390)
(0.350000000000000000000000000000,0.895067887676412469163267754518)
(0.400000000000000000000000000000,0.891943830133683107564022406465)
(0.450000000000000000000000000000,0.889418045086142244759358477901)
(0.500000000000000000000000000000,0.887397223571177386673811477469)
(0.550000000000000000000000000000,0.885804067184606280513799987061)
(0.650000000000000000000000000000,0.883653744943505243510026295446)
(0.700000000000000000000000000000,0.882997404885714031473655992086)
(0.750000000000000000000000000000,0.882566936486409509541659742082)
(0.800000000000000000000000000000,0.882329924784760282883889850516)
(0.850000000000000000000000000000,0.882258795232358792385530800990)
(0.900000000000000000000000000000,0.882330011880917930272365871855)
(0.950000000000000000000000000000,0.882523423172025141428860491183)
(1.00000000000000000000000000000,0.882821725400837315601008625440)
(1.10000000000000000000000000000,0.883675450161524799600782947538)
(1.15000000000000000000000000000,0.884206891002121765558062328976)
(1.25000000000000000000000000000,0.885430505231311325423871982887)
(1.30000000000000000000000000000,0.886107018621279261911117988345)
(1.35000000000000000000000000000,0.886817892604563725878804357468)
(1.40000000000000000000000000000,0.887557587880793346382637629218)
(1.45000000000000000000000000000,0.888321265232483452778476324861)
(1.50000000000000000000000000000,0.889104692617235512164847386368)
(1.55000000000000000000000000000,0.889904165735022149942507719459)
(1.60000000000000000000000000000,0.890716439928108397428091373378)
(1.65000000000000000000000000000,0.891538671642565068966837797725)
(1.70000000000000000000000000000,0.892368367982036654044961092696)
(1.75000000000000000000000000000,0.893203343130493105969830252032)
(1.80000000000000000000000000000,0.894041680622152998886165773979)
(1.85000000000000000000000000000,0.894881700602309402911047530287)
(1.90000000000000000000000000000,0.895721931359303675722971379841)
(1.95000000000000000000000000000,0.896561084520850360933832090646)
(2.00000000000000000000000000000,0.897398033401689536754499546955)
};

\addplot+[color=blue,mark=o,mark size=1pt,smooth]
coordinates {
(0.00100000000000000000000000000000,0.944472220667453646295428571967)
(0.100000000000000000000000000000,0.922542239551125014950191383216)
(0.150000000000000000000000000000,0.914328819363669524673741422242)
(0.200000000000000000000000000000,0.907577754533462381230089042838)
(0.250000000000000000000000000000,0.902031584746542407341423024270)
(0.300000000000000000000000000000,0.897484065526387137030815782438)
(0.350000000000000000000000000000,0.893768659399766322834821693178)
(0.400000000000000000000000000000,0.890749906383015167327390016438)
(0.450000000000000000000000000000,0.888316883287281092172020906988)
(0.500000000000000000000000000000,0.886378196599247205115597747625)
(0.550000000000000000000000000000,0.884858113817845180019664503954)
(0.600000000000000000000000000000,0.883693548655154242786464764187)
(0.650000000000000000000000000000,0.882831692803860558464147532640)
(0.700000000000000000000000000000,0.882228141686439076160968498651)
(0.750000000000000000000000000000,0.881845400770240547019317405723)
(0.800000000000000000000000000000,0.881651687371938809376741872972)
(0.850000000000000000000000000000,0.881619963582760707858832950520)
(0.900000000000000000000000000000,0.881727151218989857831114430527)
(0.950000000000000000000000000000,0.881953491065833819659411666221)
(1.00000000000000000000000000000,0.882282017207608745870463026631)
(1.05000000000000000000000000000,0.882698123682460428803014442677)
(1.10000000000000000000000000000,0.883189205608825130076891639856)
(1.15000000000000000000000000000,0.883744360695882533900119626810)
(1.20000000000000000000000000000,0.884354139957093658283918467118)
(1.25000000000000000000000000000,0.885010338704304829875712184243)
(1.30000000000000000000000000000,0.885705820664996110704013628724)
(1.35000000000000000000000000000,0.886434369452668398875000469000)
(1.40000000000000000000000000000,0.887190562716833019288586894247)
(1.45000000000000000000000000000,0.887969665170092023487308799993)
(1.50000000000000000000000000000,0.888767537385152807194195597005)
(1.55000000000000000000000000000,0.889580557812360428087561333956)
(1.60000000000000000000000000000,0.890405555917713532751285379757)
(1.65000000000000000000000000000,0.891239754704974680115344818180)
(1.70000000000000000000000000000,0.892080721180972233722817622177)
(1.75000000000000000000000000000,0.892926323564249676932278888382)
(1.80000000000000000000000000000,0.893774694234625972432594702025)
(1.85000000000000000000000000000,0.894624197583485044819105838042)
(1.90000000000000000000000000000,0.895473402058443979695061677639)
(1.95000000000000000000000000000,0.896321055806811660081787092035)
(2.00000000000000000000000000000,0.897166065414218297711721533536)
    };
\addplot+[color=purple,mark=o,mark size=1pt,smooth]
coordinates {
(0.100000000000000000000000000000,0.922354827984069851995927885282)
(0.150000000000000000000000000000,0.914160613322772346319043554848)
(0.200000000000000000000000000000,0.907425926571252907208096455525)
(0.250000000000000000000000000000,0.901893831292527958258737908174)
(0.300000000000000000000000000000,0.897358492490530533179551923850)
(0.350000000000000000000000000000,0.893653696127356710604300471598)
(0.400000000000000000000000000000,0.890644240176411064300708147979)
(0.450000000000000000000000000000,0.888219409038666362483330050504)
(0.500000000000000000000000000000,0.886287977657594260455298996379)
(0.550000000000000000000000000000,0.884774351299934986022450174584)
(0.600000000000000000000000000000,0.883615557168604327894037454955)
(0.650000000000000000000000000000,0.882758881085772746776203281010)
(0.700000000000000000000000000000,0.882159997044285953011215203689)
(0.750000000000000000000000000000,0.881781476488523808579456341350)
(0.800000000000000000000000000000,0.881591592450455314682600824827)
(0.850000000000000000000000000000,0.881563354321557884398665791476)
(0.900000000000000000000000000000,0.881673724276213429170251480687)
(0.950000000000000000000000000000,0.881902977698127155496080370334)
(1.00000000000000000000000000000,0.882234178465917649943368109012)
(1.05000000000000000000000000000,0.882652746384321510666587408522)
(1.10000000000000000000000000000,0.883146098945287697696152828165)
(1.15000000000000000000000000000,0.883703353359911764942513693346)
(1.20000000000000000000000000000,0.884315077702665906617341684191)
(1.25000000000000000000000000000,0.884973082262942850168239630864)
(1.30000000000000000000000000000,0.885670243960327141712043130074)
(1.35000000000000000000000000000,0.886400358064565549574054714904)
(1.40000000000000000000000000000,0.887158012555458297927746267988)
(1.45000000000000000000000000000,0.887938481327183082938681367356)
(1.50000000000000000000000000000,0.888737633135556704550155931138)
(1.55000000000000000000000000000,0.889551853743402652687788244083)
(1.60000000000000000000000000000,0.890377979167720965532017873917)
(1.65000000000000000000000000000,0.891213238295319696544382494930)
(1.70000000000000000000000000000,0.892055203428507060654967691069)
(1.75000000000000000000000000000,0.892901747563061758955769590201)
(1.80000000000000000000000000000,0.893751007397750630475352619412)
(1.85000000000000000000000000000,0.894601351236627959930705608860)
(1.90000000000000000000000000000,0.895451351078946348989601137467)
(1.95000000000000000000000000000,0.896299758302078001988547542607)
(2.00000000000000000000000000000,0.897145482434655238504496155934)

    };
\legend{$x=5$, $x=10$,$x=100$,$x=1000$, $x=5000$ 
}
\end{axis}
\end{tikzpicture}
\caption{\label{fig:exp} Normalized (by $k=3$) random access expectation $\E[\tau_F(\mG_{x,y})]$ from the formula in Corollary~\ref{cor:expconstrm} for various $x$ and multiplicities of the fundamental points $y$.}
\end{figure}

\section{On a rate $1/2$ construction for any dimension} \label{sec:conj}

The main goal of this section is to prove a conjecture stated in~\cite{bar2023cover} concerning a class of codes of rate $1/2$ that perform well in terms of the random access problem. We start by providing the needed definitions, preliminary results and the said conjecture.

\begin{theorem}[\textnormal{\cite[Theorem 1]{gruica2024reducing}}] \label{thm:sumra}
For a multiset of points $\mG=\{P_1,\dots,P_n\} \subseteq \mathrm{PG}(k-1,q)$ of rank $k$ we have 
$$\sum_{P \in \mG}\E\left[ {\tau}_P(\mG)\right] = kn.$$
\end{theorem}

Inspired by Theorem~\ref{thm:sumra}, in order to construct ``good'' codes for the random access problem, one could try to construct sets of points $\mG$ for which the fundamental points are balanced, and the non-fundamental points are balanced, and for which ${\alpha}_F(\mG,s) \ge  {\alpha}_N(\mG,s)$ for all $1 \le s\le n-1$ and ${\alpha}_F(\mG,s) >  {\alpha}_N(\mG,s)$ where the inequality is strict for some values of $s$. Then, since on average, recovering any of the $n$ points is $k$, and it is ``easier'' to recover fundamental points than non-fundamental points, clearly the random access expectation for fundamental points has to be smaller than $k$. This idea is formalized and proved in the following lemma, which is one of the main tools for proving the conjecture.

\begin{lemma} \label{lem:delta}
Let $\mG=\{P_1,\dots,P_n\}\subseteq \mathrm{PG}(k-1,q)$ be a set of distinct points in $\F_q^k$ of rank~$k$ that contains all fundamental points from $\mathrm{PG}(k-1,q)$, and that is balanced within the fundamental part, and within the non-fundamental part. Let $d$ denote the minimum distance of the code arising from the point set $\mG$. Let
\begin{align*}
    \delta := \displaystyle\sum_{s=1}^{n-d}\frac{\alpha_F(\mG,s)-\alpha_N(\mG,s)}{\binom{n-1}{s}} 
\end{align*}
then $\E[\tau_F(\mG)] = k-\frac{(n-k)}{n}\delta$.
\end{lemma}
\begin{proof}
From Lemma~\ref{lem:fi} it follows that we have
\begin{align*}
    \E[\tau_F(\mG)] + \delta &= nH_n-\sum_{s=1}^{n-1}\frac{\alpha_F(\mG,s)}{\binom{n-1}{s}} +\displaystyle\sum_{s=1}^{n-d}\frac{\alpha_F(\mG,s)-\alpha_N(\mG,s)}{\binom{n-1}{s}} \\
    &= nH_n-\sum_{s=1}^{n-1}\frac{\alpha_F(\mG,s)}{\binom{n-1}{s}} +\displaystyle\sum_{s=1}^{n-1}\frac{\alpha_F(\mG,s)-\alpha_N(\mG,s)}{\binom{n-1}{s}} \\
    &= nH_n - \sum_{s=1}^{n-1}\frac{\alpha_N(\mG,s)}{\binom{n-1}{s}} \\
    &= \E[\tau_N(\mG)]
\end{align*}
where in the first equality we used the fact that $\alpha_F(\mG,s)=\alpha_N(\mG,s)$ for $n-d+1 \le s \le n-1$, which holds because a code of minimum distance $d$ can correct $d-1$ erasures, and so any $n-d+1$ or more points in $\mG$ will recover the remaining points.
Combining this with the statement of Theorem~\ref{thm:sumra} we obtain
\begin{align*}
    kn &= \sum_{P \in \mG}\E\left[ {\tau}_P(\mG)\right] 
 \\
 &= \sum_{\substack{P \in \mG \\ \textnormal{$P$ is fundamental}}}\E\left[ {\tau}_F(\mG)\right] + \sum_{\substack{P \in \mG \\ \textnormal{$P$ is non-fundamental}}}\E\left[ {\tau}_N(\mG)\right] \\
    &= k\E\left[ {\tau}_F(\mG)\right] + (n-k) \E\left[ {\tau}_N(\mG)\right] \\
    &= k\E[\tau_F(\mG)]+(n-k)\left(\E[\tau_F(\mG)]+\delta \right) \\
    &= n\E[\tau_F(\mG)]+(n-k)\delta.
\end{align*}
We therefore have $\E[\tau_F(\mG)] = k-\frac{(n-k)}{n}\delta$.
\end{proof}

We give an example illustrating how one can apply Lemma~\ref{lem:delta}.

\begin{example} \label{ex:lines}
    Consider the set of points
    \begin{align*} 
    \mG= \{ &(1:0:0:0), (0:1:0:0), (0:0:1:0), (0:0:0:1), \\ 
    &(1:1:0:0),(0:1:1:0),(0:0:1:1),(1:0:0:1)\}
    \end{align*}
    in $\F_2^4$ of rank 4.
Then one can check that
\begin{align*}
    ({\alpha}_F(\mG,1), {\alpha}_F(\mG,2), \dots, {\alpha}_F(\mG,7)) &= (1,9,33,62,55,28,8) \\
    ({\alpha}_N(\mG,1), {\alpha}_N(\mG,2), \dots, {\alpha}_N(\mG,7))  &= (1,8,29,58,54,28,8),
\end{align*}
and so
\begin{align*}
    \delta = \frac{1}{\binom{7}{2}}+\frac{4}{\binom{7}{3}}+\frac{4}{\binom{7}{4}}+\frac{1}{\binom{7}{5}} = \frac{34}{105}.
\end{align*}
Therefore, by Lemma~\ref{lem:delta}
we conclude that $\E[\tau_F(\mG)]= 4-\frac{(8-4)34}{8\cdot 105}  = \frac{403}{105} \approx 3.838 < 4$.
\end{example}

The class of codes that we study in this section was first introduced in~\cite{bar2023cover}, and in the language of this paper, is defined as follows.

\begin{definition}[\textnormal{\cite[Construction 1]{bar2023cover}}] \label{def:daniella}
    Let $E_i$ denote the $i$-th fundamental point in \smash{$\mathrm{PG}(k-1,q)$} and let $E_{i,j}$ denote the point corresponding to the sum $E_i+E_j$. We consider the point set $\mG=\mF \cup \mN$ where $\mF=\{E_1,\dots,E_k\}$ is the set of all fundamental points in~$\mathrm{PG}(k-1,q)$, and $\mN=\{E_{1,2},E_{2,3},\dots,E_{k-1,k},E_{k,1}\}$ is a set of non-fundamental points obtained by summing consecutive fundamental points. 
\end{definition}

To be able to refer to this point set from Definition~\ref{def:daniella} later, throughout this section we denote it by $\mG_k$. Also, in the sequel, we use the Bachmann-Landau notation ``$\sim$'' to describe the asymptotic growth of real-valued functions defined on an infinite set of natural numbers; see e.g.~\cite{de1981asymptotic}.

A nice property of the point set $\mG_k$ is that the code arising from it has rate $1/2$ for any value of $k$, and thus, compared to other constructions of codes that perform well in terms of the random access problem, does not have vanishing rate. In~\cite{bar2023cover} it was conjectured that when $q=2$ this code has a relatively small random access expectation with respect to its dimension~$k$. More precisely, the following conjecture was proposed.

\begin{conjecture}[\textnormal{\cite[Conjecture 1]{bar2023cover}}] \label{con:daniella}
Let $\mG_k$ be defined as in Definition~\ref{def:daniella}. The ratio $\E[\tau_F(\mG_k)]/k$ decreases with $k$, and we have
\begin{align*}
    \lim_{k \to \infty} \frac{\E[\tau_F(\mG_k)]}{k} < 0.9456.
\end{align*}
\end{conjecture}

Furthermore, in~\cite[Example 2]{bar2023cover} it was shown that $\E[\tau_F(\mG_4)]/4 \approx 3.838 = 0.959k$. As for $k > 4$, coming up with explicit expressions for the parameters needed for computing $\E[\tau_F(\mG_k)]$ seems to be a hard problem, the authors came up with a second-order recursive formula for some code parameters that are closely related to $\alpha_F(\mG_k,s)$. Inspired by Lemma~\ref{lem:delta}, we give a first-order recursion for the $\alpha_F(\mG_k,s)$, as well as for $\alpha_N(\mG_k,s)$, and give bounds on their difference. This will in turn lead to an upper bound for $\E[\tau_F(\mG_k)]$ from Lemma~\ref{lem:delta}. Note that we use the same approach as was done in~\cite{bar2023cover}.

In the sequel, we denote $\beta_F(\mG_k,s) :=\binom{2k}{s}-\alpha_F(\mG_k,s)$ and $\beta_N(\mG_k,s) :=\binom{2k}{s}-\alpha_N(\mG_k,s)$ for all $1 \le s \le 2k-1.$ Note that contrary to $\alpha_F(\mG_k,s)$ and $\alpha_N(\mG_k,s)$, $\beta_F(\mG_k,s)$ and $\beta_N(\mG_k,s)$, respectively, count the number of subsets of cardinality $s$ in $\mG_k$ that \emph{do not} recover fundamental points, and non-fundamental points, respectively.

\begin{proposition} \label{prop:recF}
We have
\begin{align*}
    \beta_F(\mG_k,s) = \begin{cases}
        0 \quad &\textnormal{if $k=1$, $s \ge 1$} \\
        2k-1 \quad  &\textnormal{if $k > 1$, $s =1$} \\
        1 \quad  &\textnormal{if $k=2$, $s=2$} \\
        0 \quad  &\textnormal{if $k=2$, $s=3$} \\
        \beta_F(\mG_{k-1},s-1)+\displaystyle\sum_{j=0}^{s}\binom{2k-2j-3}{s-j} +1 \quad  &\textnormal{if $k \ge 3$, $s=k-1$} \\
        \beta_F(\mG_{k-1},s-1)+\displaystyle\sum_{j=0}^{s}\binom{2k-2j-3}{s-j}  \quad  &\textnormal{if $k \ge 3$, $2 \le s \le 2k-3$, \newline $s \ne k-1$} \\
        0 \quad  &\textnormal{if $k \ge 3$, $s\ge 2k-2$} \\
    \end{cases}
\end{align*}
\end{proposition}
\begin{proof}
We compute a recursive formula for the number of $s$-sets of points in $\mG_k$ that recover $E_1$ by taking the same approach to the one in the proof of~\cite[Theorem 10]{bar2023cover}.
Consider the directed graph $G_k$ whose vertices are made of the $2k-1$ points in $\mG_k \setminus E_1$, and whose edges can be partitioned into \emph{blue} and \emph{red} edges, which are characterized as follows.
\begin{itemize}
    \item[(i)] For every $1 \le i \le k-2$ there are two outgoing blue edges from $E_{i,i+1}$: one to the vertex $E_{i+1}$ and one to the vertex $E_{i+1,i+2}$;
    \item[(ii)] There are two outgoing blue edges from $E_{k-1,k}$: one to $E_k$ and one to $E_{k,1}$;
    \item[(iii)] For every $2 \le i \le k-1$ there are two outgoing red edges from $E_{i,i+1}$: one to the vertex $E_{i}$ and one to the vertex $E_{i-1,i}$;
    \item[(iv)] There are two outgoing red edges from $E_{k,1}$: one to $E_k$ and one to $E_{k-1,k}$.
\end{itemize}
A depiction of this graph is given in Figure~\ref{fig:rece1}.
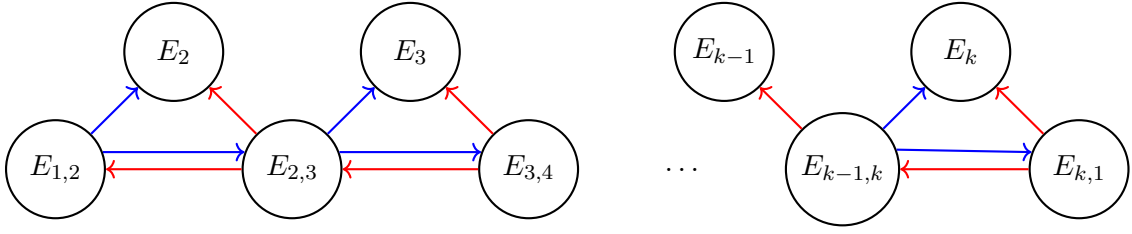
\begin{figure}[ht!]
\begin{adjustwidth*}{}{}
\begin{center}
\begin{tikzpicture}[node distance={22mm}, thick, main/.style = {draw, circle, minimum size=1.3cm}] 
\node[main] (1) {$E_{1,2}$}; 
 \node[main] (2) [above right of=1] {$E_2$}; 
 \node[main] (3) [below right of=2] {$E_{2,3}$};
  \node[main] (4) [above right of=3] {$E_3$}; 
 \node[main] (5) [below right of=4] {$E_{3,4}$};
  \node[main] (7) [right=2.8cm of 4]{$E_{k-1}$};
   \node[main] (6) [below right of=7]{$E_{k-1,k}$};
     \node[main] (8) [above right of=6] {$E_k$};
      \node[main] (9) [below right of=8] {$E_{k,1}$};
\path (5) -- node[auto=false]{$\ldots$} (6);
\draw[->,blue] (1) -- (2);
\draw[->,blue] (3) -- (4);
\draw[->,blue] (6) -- (8);
\draw[->,red] (9) -- (8);
\draw[->,red] (6) -- (7);
\draw[->,red] (5) -- (4);
\draw[->,red] (3) -- (2);

\draw[->,blue] (1.20) -- (3.160);
\draw[->,red] (3) -- (1);

\draw[->,blue] (3.20) -- (5.160);
\draw[->,red] (5) -- (3);

\draw[->,blue] (6.20) -- (9.160);
\draw[->,red] (9) -- (6);
\end{tikzpicture} 
\end{center}
\end{adjustwidth*}
\caption{\label{fig:rece1} The graph for recovering $E_1$. The notation $E_{i,j}$ indicates the sum $E_i+E_j$.}
\end{figure}

Note that there is a one-to-one correspondence between $s$-sets of points in $\mG_k$ that do \emph{not} recover $E_1$ and subsets of vertices in $G_k$ that do not contain a monochromatic path starting from either $E_{1,2}$ or $E_{k,1}$ and ending in a vertex $E_i$ for some $2 \le i \le k$. We split the proof into three parts and we denote by $V \subseteq G_k$ a subgraph of $G_k$ made of $s$ vertices that does not constitute a set of points that recover $E_1$.
\begin{enumerate}
    \item If neither $E_{1,2}$ nor $E_{k,1}$ are in $V$, then the points in $V$ will never recover $E_1$, and so there are $\binom{2k-3}{s}$ such subgraphs.
    \item If $E_{1,2}\in V$, then $E_2 \notin V$. Moreover, the number of such subgraphs $V$ is exactly $\beta_F(\mG_{k-1},s-1)$.
    \item If $E_{1,2}\notin V$ and $E_{k,1} \in V$ then $E_{k} \notin V$ and $V$ cannot contain a red path starting in $E_{k-1,k}$ and ending in a vertex of the form $E_i$ for $2 \le i \le k$. The number of subgraphs $V$ that do not contain such a path can be counted in the following way:
    \begin{enumerate}
        \item[3.1] If $E_{k-1,k} \notin V$, then we can choose any of the remaining vertices and will not end up with a red path in $V$ starting at $E_{k-1,k}$. There are $\binom{2k-5}{s-1}$ such subgraphs.
        \item[3.2] If $E_{k-1,k} \in V$, then $E_{k-1} \notin V$. If $E_{k-2,k-1} \notin V$ then we can choose any of the remaining vertices and will not end up with a red path in $V$ starting at $E_{k-1,k}$. There are $\binom{2k-7}{s-2}$ such subgraphs.
        \item[3.3] If $E_{k-1,k},E_{k-2,k-1} \in V$, then $E_{k-1},E_{k-2} \notin V$. If $E_{k-3,k-2} \notin V$ then we can choose any of the remaining vertices and will not end up with a red path in $V$ starting at $E_{k-1,k}$. There are $\binom{2k-9}{s-3}$ such subgraphs.
        \item[\vdots]
    \end{enumerate}
    Proceeding with the steps above we end up with a total of $\sum_{j=1}^{s}\binom{2k-2j-3}{s-j}$ subgraphs that will not recover $E_1$. Finally, note that in the special case where $s=k-1$ we did not count the subgraph that consists of the vertices in $\{E_{2,3},E_{3,4},\dots,E_{k-1,k},E_{k,1}\}$. 
    
    Moreover, when $s \ge 2k-2$, then any subset of size $s$ will recover $E_1$, and the other cases of the proposition follow directly from the definitions. \qedhere
\end{enumerate}
\end{proof}

Unlike the case of fundamental points, for the non-fundamental points we were not able to find a recursive formula, but we were able to provide a lower bound on $ \beta_N(\mG_k,s)$ depending also on $ \beta_N(\mG_{k-1},s-1)$. Note it is very likely that one can find an exact recursive formula for $ \beta_N(\mG_k,s)$, but this would require a more detailed case analysis. For the purpose of this paper, however, our lower bound suffices.

\begin{proposition} \label{prop:recN}
We have
\begin{align*}
    \beta_N(\mG_k,s) \ge \begin{cases}
        0 \quad &\textnormal{if $k=1$, $s \ge 1$} \\
        2 \quad  &\textnormal{if $k=2$, $s =1$} \\
        2k-1 \quad  &\textnormal{if $k\ge 3$, $s=1$} \\
        0 \quad  &\textnormal{if $k=2$, $s\ge 2$} \\
        8 \quad  &\textnormal{if $k=3$, $s= 2$} \\
        2 \quad  &\textnormal{if $k=3$, $s=3$} \\
        \binom{2k-3}{s}+\beta_N(\mG_{k-1},s-1) +\\ \displaystyle\sum_{j=0}^{s-1}\binom{2k-2j-5}{s-j-1}+\displaystyle\sum_{j=0}^{s-2}\binom{2k-2j-7}{s-j-2}+ \\ 2\displaystyle\sum_{j=0}^{s-3}\binom{2k-2j-7}{s-j-3}+ \displaystyle\sum_{j=0}^{s-4}\binom{2k-2j-7}{s-j-4}
        \quad  &\textnormal{if $k \ge 4$, $2 \le s \le 2k-3$} \\
        0 \quad  &\textnormal{if $k \ge 3$, $s\ge 2k-2$} \\
    \end{cases}
\end{align*}
\end{proposition}
\begin{proof}
We compute a recursive formula by counting the number of $s$-sets of points in $\mG_k$ that recover $E_{1,2}$. Similarly to the proof of Proposition~\ref{prop:recF}, the arguments we use are based on counting certain paths in an appropriate directed graph: consider the directed graph $G_k$ whose vertices are made of the $2k-1$ points in $\mG_k \setminus E_{1,2}$, and whose set of edges consists of \emph{blue} and \emph{red} edges, defined as follows.
\begin{itemize}
    \item[(i)] For every $2 \le i \le k-2$ there are two outgoing blue edges from $E_{i,i+1}$: one to the vertex $E_{i+1}$ and one to the vertex $E_{i+1,i+2}$;
    \item[(ii)] There are two outgoing blue edges from $E_{k-1,k}$: one to $E_k$ and one to $E_{k,1}$;
    \item[(iii)] For every $3 \le i \le k$ there is an outgoing blue edge from $E_i$ to $E_1$;
    \item[(iv)] For every $2 \le i \le k-1$ there are two outgoing red edges from $E_{i,i+1}$: one to the vertex $E_{i}$ and one to the vertex $E_{i-1,i}$;
    \item[(v)] There are two outgoing red edges from $E_{k,1}$: one to $E_k$ and one to $E_{k-1,k}$;
    \item[(vi)] For every $3 \le i \le k$ there is an outgoing red edge from $E_i$ to $E_2$.
\end{itemize}
A depiction of this graph is given in Figure~\ref{fig:rece12}. 
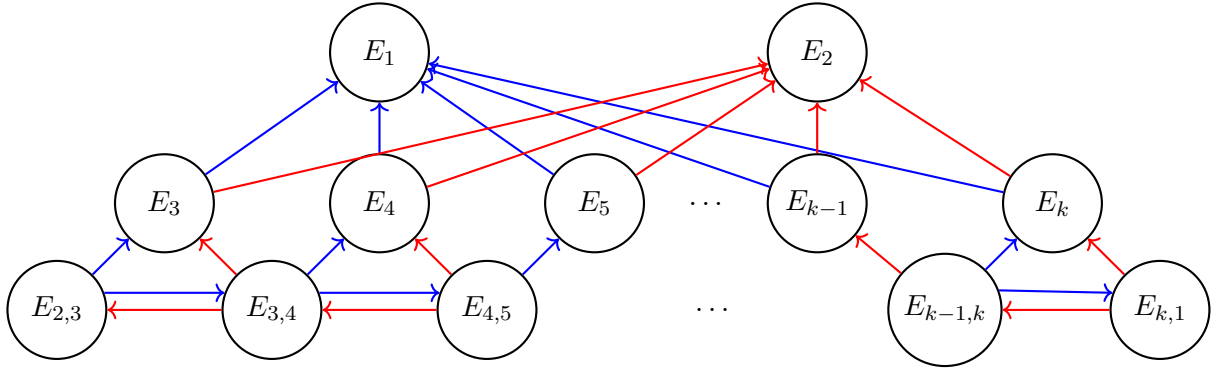
\begin{figure}[ht!]
\begin{adjustwidth*}{}{}
\begin{center}
\begin{tikzpicture}[node distance={20mm}, thick, main/.style = {draw, circle, minimum size=1.3cm}] 
\node[main] (1) {$E_{2,3}$}; 
\node[main] (2) [above right of=1] {$E_3$}; 
\node[main] (3) [below right of=2] {$E_{3,4}$}; 
\node[main] (4) [above right of=3] {$E_4$}; 
\node[main] (5) [below right of=4] {$E_{4,5}$}; 
\node[main] (6) [above right of=5] {$E_5$}; 
\node[main] (7) [right=4.6cm of 5] {$E_{{k-1,k}}$}; 
\node[main] (8) [right=1.6cm of 6] {$E_{k-1}$}; 
\node[main] (9) [above right of=7] {$E_{k}$}; 
\node[main] (10) [below right of=9] {$E_{k,1}$};
\node[main] (11) [above of=4] {$E_1$};
\node[main] (12) [above of=8] {$E_2$};

\path (6) -- node[auto=false]{$\ldots$} (8);
\path (5) -- node[auto=false]{$\ldots$} (7);
\draw[->,blue] (1.20) -- (3.160);
\draw[->,red] (3) -- (1);
\draw[->,blue] (3.20) -- (5.160);
\draw[->,red] (5) -- (3);
\draw[->,blue] (7.20) -- (10.160);
\draw[->,red] (10) -- (7);

\draw[->,blue] (1) -- (2);
\draw[->,blue] (3) -- (4);
\draw[->,blue] (5) -- (6);
\draw[->,blue] (7) -- (9);

\draw[->,red] (3) -- (2);
\draw[->,red] (5) -- (4);
\draw[->,red] (7) -- (8);
\draw[->,red] (10) -- (9);

\draw[->,blue] (2) -- (11);
\draw[->,blue] (4) -- (11);
\draw[->,blue] (6) -- (11);
\draw[->,blue] (9) -- (11);
\draw[->,blue] (8) --(11);

\draw[->,red] (2) -- (12);
\draw[->,red] (4) -- (12);
\draw[->,red] (6) -- (12);
\draw[->,red] (9) -- (12);
\draw[->,red] (8) --(12);
\end{tikzpicture}
\end{center}
\end{adjustwidth*}
\caption{\label{fig:rece12} The graph for recovering $E_1+E_2$. The notation $E_{i,j}$ indicates again the sum $E_i+E_j$.}
\end{figure}
There is a one-to-one correspondence between $s$-sets of points in $\mG_k$ that do \emph{not} recover $E_{1,2}$ and subsets of vertices in $G_k$ that do not contain a monochromatic path starting from either $E_{2,3}$ or $E_{k,1}$ and ending in either $E_1$ or $E_2$. As in Proposition~\ref{prop:recF}, we split the proof into different cases, and denote by $V \subseteq G_k$ a subgraph of $G_k$ made of $s$ vertices that does not constitute a set of points that recover $E_{1,2}$:
\begin{enumerate}
    \item If $E_{2,3}, E_2 \notin V$, then there is no way to recover $E_{1,2}$ from any set of vertices, and so we have $\binom{2k-3}{s}$ such subgraphs.
    \item If $E_{2,3} \notin V$ and $E_2 \in V$, then $E_1 \notin V$. Moreover, $V$ cannot contain a red path starting in $E_{k,1}$ and ending in one of the vertices $E_i$ for $3 \le i \le k$. We have counted the number of such subgraphs in part 3. of the proof of Proposition~\ref{prop:recF}, and it is equal to $\sum_{j=0}^{s-1}\binom{2k-2j-5}{s-j-1}$. Note that we again have a special case when $s=k-1$, however, since we only give a lower bound for $\beta(\mG_k,s)$, we do not treat this extra case.
    \item If $E_{2,3} \in V$ and $E_3 \notin V$, then the number of such subgraphs is $\beta_N(\mG_{k-1},s-1)$.
    \item If $E_{2,3},E_3 \in V$ then $E_1 \notin V$. Assume also that $E_{3,4} \notin V$. Then $V$ cannot contain a red subgraph starting at $E_{k,1}$ and ending in $E_i$ for $4 \le i \le k$. Again we use the count from the proof of Proposition~\ref{prop:recF} and obtain that there are $\sum_{j=0}^{s-2}\binom{2k-2j-7}{s-j-2}$ such paths if $E_2 \notin V$, and there are $\sum_{j=0}^{s-3}\binom{2k-2j-7}{s-j-3}$ such paths if $E_2 \in V$ as well. 
    \item If $E_{2,3},E_3,E_{3,4} \in V$ then also $E_1 \notin V$. Using analogous reasoning as before we obtain that the number of such subgraphs is $\sum_{j=0}^{s-3}\binom{2k-2j-7}{s-j-3}$ such paths if $E_2 \notin V$, and there are $\sum_{j=0}^{s-4}\binom{2k-2j-7}{s-j-4}$ if $E_2 \in V$.
\end{enumerate}
Finally, using the definitions of $\mG_k$ one can compute the initial conditions for the recursion, ending up with the statement of the proposition.
\end{proof}

Now that we have recursive expressions and bounds for both $\beta_F(\mG_k,s)$ and $\beta_N(\mG_k,s)$, we can give a lower bound for the value of $\delta$ described in Lemma~\ref{lem:delta}, that will later give an upper bound for $\E[\tau_F(\mG_k)]$ that is strictly smaller than $k$. Note that the main result of this section is a proof for Conjecture~\ref{con:daniella}, which concerns an asymptotic upper bound as $k \to \infty$.

\begin{theorem} \label{thm:bb}
Let $\delta_k:=\displaystyle\sum_{s=1}^{2k-3}\frac{\alpha_F(\mG_k,s)-\alpha_N(\mG_k,s)}{\binom{2k-1}{s}} $. We have $\lim_{k\to \infty} \frac{\delta_k}{k} \ge  \frac{4045443}{37182145}$.
\end{theorem}
\begin{proof}
Since $\alpha_F(\mG_k,s)=\binom{2k}{s}-\beta_F(\mG_k,s)$ and  $\alpha_N(\mG_k,s)=\binom{2k}{s}-\beta_N(\mG_k,s)$, from the formulas in Propositions~\ref{prop:recF} and~\ref{prop:recN} we obtain that for $k \ge 4$ and $2 \le s \le 2k-3$ we have
\begin{align*}
    \alpha_F(\mG_k,s)-\alpha_N(\mG_k,s) = &\beta_N(\mG_k,s)-\beta_F(\mG_k,s) \\
    \ge &\beta_N(\mG_{k-1},s-1)-\beta_F(\mG_{k-1},s-1) +\binom{2k-3}{s}+\displaystyle\sum_{j=0}^{s-1}\binom{2k-2j-5}{s-j-1} \\ &+ \displaystyle\sum_{j=0}^{s-2}\binom{2k-2j-7}{s-j-2} + 2\displaystyle\sum_{j=0}^{s-3}\binom{2k-2j-7}{s-j-3}+ \displaystyle\sum_{j=0}^{s-4}\binom{2k-2j-7}{s-j-4}\\ &-\displaystyle\sum_{j=0}^{s}\binom{2k-2j-3}{s-j} -1 \\
    = &\beta_N(\mG_{k-1},s-1)-\beta_F(\mG_{k-1},s-1) + \displaystyle\sum_{j=0}^{s-2}\binom{2k-2j-7}{s-j-2} \\
    &+ 2\displaystyle\sum_{j=0}^{s-3}\binom{2k-2j-7}{s-j-3}+ \displaystyle\sum_{j=0}^{s-4}\binom{2k-2j-7}{s-j-4} -1 \\
    \ge &\beta_N(\mG_{k-1},s-1)-\beta_F(\mG_{k-1},s-1) + \displaystyle\sum_{j=0}^{s-3}\binom{2k-2j-6}{s-j-2} \\
    &+ \displaystyle\sum_{j=0}^{s-4}\binom{2k-2j-6}{s-j-3} \\
    \ge &\beta_N(\mG_{k-1},s-1)-\beta_F(\mG_{k-1},s-1) +\displaystyle\sum_{j=0}^{s-4}\binom{2k-2j-5}{s-j-2} \\
\end{align*}
where we used the well-known identity for binomial coefficients stating that for integers satisfying $a > b \ge 0$ we have $\binom{a}{b}=\binom{a-1}{b-1}+\binom{a-1}{b}$. Note that the inequality above implies by induction that $\beta_N(\mG_k,s)-\beta_F(\mG_k,s) \ge 0$. Recalling that our aim is to estimate the ratio $\delta_k/k$ as $k$ goes to infinity, we obtain from the above inequality that for large values of $k$,
\begin{align} \label{eq:ess}
    \delta_k &\ge \sum_{s=1}^{2k-3} \frac{\beta_N(\mG_{k-1},s-1)-\beta_F(\mG_{k-1},s-1) + \sum_{j=0}^{s-4}\binom{2k-2j-5}{s-j-2} }{\binom{2k-1}{s}} \nonumber \\
    &= \sum_{s=1}^{2k-3} \frac{\beta_N(\mG_{k-1},s-1)-\beta_F(\mG_{k-1},s-1)}{\binom{2k-1}{s}} +\sum_{s=1}^{2k-3} \frac{ \sum_{j=0}^{s-4}\binom{2k-2j-5}{s-j-2} }{\binom{2k-1}{s}}  \\
    &\ge \sum_{s=1}^{2k-3} \frac{\beta_N(\mG_{k-1},s-1)-\beta_F(\mG_{k-1},s-1)}{\binom{2k-1}{s}} +
    \sum_{s=13}^{2k-3} \frac{ \sum_{j=0}^{9}\binom{2k-2j-5}{s-j-2} }{\binom{2k-1}{s}}. \nonumber
\end{align}
We have the following asymptotic estimates for $\sum_{s=13}^{2k-3}{\binom{2k-2j-5}{s-j-2}}/{\binom{2k-1}{s}}$ as $k \to \infty$  for each~$j \in [9]$, which we prove in Appendix~\ref{sec:apen}.
\begin{restatable}{claim}{clest} \label{cl:est}
We have
\begin{align*}
     &\sum_{s=13}^{2k-3}\frac{\binom{2k-5}{s-2}}{\binom{2k-1}{s}}  \sim \frac{k}{15}, \quad \sum_{s=13}^{2k-3}\frac{\binom{2k-7}{s-3}}{\binom{2k-1}{s}}  \sim \frac{k}{70}, \quad 
     \sum_{s=13}^{2k-3}\frac{\binom{2k-9}{s-4}}{\binom{2k-1}{s}}  \sim \frac{k}{315}, \quad \sum_{s=13}^{2k-3}\frac{\binom{2k-11}{s-5}}{\binom{2k-1}{s}}  \sim \frac{k}{1386}, \\
     &\sum_{s=13}^{2k-3}\frac{\binom{2k-13}{s-6}}{\binom{2k-1}{s}}  \sim \frac{k}{6006}, \quad \sum_{s=13}^{2k-3}\frac{\binom{2k-15}{s-7}}{\binom{2k-1}{s}}  \sim \frac{k}{25740}, \quad 
     \sum_{s=13}^{2k-3}\frac{\binom{2k-17}{s-8}}{\binom{2k-1}{s}}  \sim \frac{k}{109395},  \\
     &\sum_{s=13}^{2k-3}\frac{\binom{2k-19}{s-9}}{\binom{2k-1}{s}}  \sim \frac{k}{461890}, \quad
     \sum_{s=13}^{2k-3}\frac{\binom{2k-21}{s-10}}{\binom{2k-1}{s}}  \sim \frac{k}{1939938}, \quad \sum_{s=13}^{2k-3}\frac{\binom{2k-23}{s-11}}{\binom{2k-1}{s}}  \sim \frac{k}{8112468}
\end{align*}
as $k \to \infty$.
\end{restatable}
From Claim~\ref{cl:est} we obtain  
\begin{align*}
    \lim_{k \to \infty} \frac{\delta_k}{k} &\ge 
    \lim_{k \to \infty} \sum_{s=1}^{2k-3} \frac{\beta_N(\mG_{k-1},s-1)-\beta_F(\mG_{k-1},s-1)}{k\binom{2k-1}{s}}\\
    &+\frac{1}{15}+ \frac{1}{70}+\frac{1}{315}+\frac{1}{1386} +\frac{1}{6006} +\frac{1}{25740} +\frac{1}{109395}+\frac{1}{461890}+\frac{1}{1939938}+\frac{1}{8112468}.
\end{align*}
We also have 
\begin{align*}
\sum_{s=1}^{2k-3} \frac{\beta_N(\mG_{k-1},s-1)-\beta_F(\mG_{k-1},s-1)}{\binom{2k-1}{s}} 
\ge \;  &\sum_{s=1}^{2k-3} \frac{\beta_N(\mG_{k-2},s-2)-\beta_F(\mG_{k-2},s-2)}{\binom{2k-1}{s}} \\ &+ \sum_{s=12}^{2k-3} \frac{ \sum_{j=0}^{8}\binom{2k-2j-7}{s-j-3} }{\binom{2k-1}{s}}
\end{align*}
as $k \to \infty$. Again using the estimates from Claim~\ref{cl:est}, we have
\begin{align*}
\lim_{k \to \infty}\sum_{s=1}^{2k-3} &\frac{\beta_N(\mG_{k-1},s-1)-\beta_F(\mG_{k-1},s-1)}{k\binom{2k-1}{s}} \ge  \lim_{k \to \infty}\sum_{s=1}^{2k-3} \frac{\beta_N(\mG_{k-2},s-2)-\beta_F(\mG_{k-2},s-2)}{k\binom{2k-1}{s}}\\
&+ \frac{1}{70}+\frac{1}{315}+\frac{1}{1386} +\frac{1}{6006} +\frac{1}{25740} +\frac{1}{109395}+\frac{1}{461890}+\frac{1}{1939938}+\frac{1}{8112468}.
 \end{align*}
We can proceed in the same way to obtain
\begin{align*}
\lim_{k \to \infty}\sum_{s=1}^{2k-3} &\frac{\beta_N(\mG_{k-2},s-2)-\beta_F(\mG_{k-2},s-2)}{k\binom{2k-1}{s}} \ge  \lim_{k \to \infty}\sum_{s=1}^{2k-3} \frac{\beta_N(\mG_{k-3},s-3)-\beta_F(\mG_{k-3},s-3)}{k\binom{2k-1}{s}}\\
&+\frac{1}{315}+\frac{1}{1386} +\frac{1}{6006} +\frac{1}{25740} +\frac{1}{109395}+\frac{1}{461890}+\frac{1}{1939938}+\frac{1}{8112468}.
 \end{align*}
as $k \to \infty$. Going on with the same procedure seven more times, we finally obtain
\begin{align*}
    \lim_{k \to \infty} \frac{\delta_k}{k} &\ge \lim_{k \to \infty}\sum_{s=1}^{2k-3}\frac{\beta_N(\mG_{k-10},s-10)-\beta_F(\mG_{k-10},s-10)}{k\binom{2k-1}{s}} \\
    &+ \frac{1}{15}+ \frac{2}{70}+\frac{3}{315}+\frac{4}{1386} +\frac{5}{6006} +\frac{6}{25740} +\frac{7}{109395}+\frac{8}{461890}+\frac{9}{1939938}+\frac{10}{8112468} \\
    &= \lim_{k \to \infty}\sum_{s=1}^{2k-3}\frac{\beta_N(\mG_{k-10},s-10)-\beta_F(\mG_{k-10},s-10)}{k\binom{2k-1}{s}} + \frac{4045443}{37182145} \\
    &\ge \frac{4045443}{37182145}
\end{align*}
as $k \to \infty$.
Note that in the above argument we used that $\beta_N(\mG_{k-10},s-10)-\beta_F(\mG_{k-10},s-10)\geq 0$, which is true by induction as recalled at the beginning of the proof.
\end{proof}

Finally, with the asymptotic estimate in Theorem~\ref{thm:bb} we can prove Conjecture~\ref{con:daniella}.

\begin{corollary} \label{cor:proof}
We have
\begin{align*}
    \lim_{k \to \infty} \frac{\E[\tau_F(\mG_k)]}{k} \le \frac{70318847}{74364290}.
\end{align*}
\end{corollary}
\begin{proof}
From Lemma~\ref{lem:delta} and Theorem~\ref{thm:bb} we  obtain 
\begin{align*}
    \lim_{k \to \infty} \frac{\E[\tau_F(\mG_k)]}{k} \le \lim_{k \to \infty}1-\frac{\delta_k}{2k} \le 1-\frac{4045443}{2 \cdot 37182145} = \frac{70318847}{74364290}
\end{align*}
concluding the proof.
\end{proof}

Since ${70318847}/{74364290} \approx 0.945599655 < 0.9456$, Corollary~\ref{cor:proof} proves Conjecture~\ref{con:daniella}.

\section{Conclusion and Future Directions} \label{sec:concl}

In this paper, we study the Random Access Problem, which seeks to compute the expected number of samples needed to recover an information strand. Unlike previous approaches, we take a geometric perspective on the problem. This approach allowed us to identify properties of codes that lead to better performance in reducing the random access expectation. As a result, we developed a construction for $k = 3$ that outperforms prior attempts at minimizing the random access expectation. While our method could be extended to higher dimensions, the necessary computations become significantly more complex, so we leave this as an open direction for future research.

In addition to our construction, we also applied a result from~\cite{gruica2024reducing}, that implies that showing that recovering information strands is more likely than recovering non-information strands, is equivalent to proving that the random access expectation is less than $k$. This was the main tool in addressing a conjecture posed in~\cite{bar2023cover}.

We conclude the paper with some other possible directions for future research:

Even though in this paper we only investigated the expected number of draws in order to recover \emph{one} point, in practice it would be interesting to study the expected number of draws for recovering a subset of fundamental points, and determine codes for which this expectation is small (or at least smaller than what one would get using an identity code). The construction coming from balanced quasi-arcs makes it not only easy to give a closed formula for the random access expectation as treated in this paper, but also for the case of retrieving \emph{two} fundamental points. We compared our construction for this case with the expectation of MDS codes with a computer algebra program, and it outperforms MDS codes, i.e., the expected number of draws is smaller than for the codes arising from balanced quasi-arcs. We believe it would be interesting to investigate this in more detail, and to try to construct codes that reduce the random access expectation for recovering two fundamental points, possibly taking a similar geometric approach as done in this paper.

Another natural direction would be to improve the asymptotic estimates used in Corollary~\ref{cor:proof}. More precisely, in establishing the asymptotic upper bound of Corollary~\ref{cor:proof}, we applied several approximate estimates. Consequently, the actual limit of $\lim_{k \to \infty} {\E[\tau_F(\mG_k)]}/{k}$ is smaller than what was initially conjectured in~\cite{bar2023cover}. We provide explicit evaluations of ${\E[\tau_F(\mG_k)]}/{k}$ for various values of $k$ in Table~\ref{tb}, that we obtained with a computer algebra program. One can see, that already for $k=20$ we have ${\E[\tau_F(\mG_k)]}/{k} < 0.9456$ as conjectured in~\cite{bar2023cover} for $k \to \infty$. We believe that the real asymptotic upper bound should lie further below $0.9456k$.

\begin{table}[h!]
    \centering
    \begin{tabular}{|c|c|c|c|c|c|c|c|c|} 
 \hline
$k$ & 4 & 5 & 6 & 7 & 8 & 9 & 10 & 20 \\
\hline
${\E[\tau_F(\mG_k)]}/{k}$ & 0.95952 & 0.94921 & 0.94653 & 0.94584 & 0.94566 & 0.94562 & 0.94560 & 0.94559 \\
\hline 
\end{tabular}
\caption{${\E[\tau_F(\mG_k)]}/{k}$ for various values of $k$.}
\label{tb}
\end{table}

One way of refining the upper bound of $0.9456k$ would be to evaluate more values for $j$ in the sum beyond the 10 currently considered in the estimation right after~\eqref{eq:ess}, and getting a tighter lower bound on $\delta_k$ (see Theorem~\ref{thm:bb}). Additionally, extending the recursion to consider more than 10 previous steps could also contribute to a tighter upper bound.


\section*{Acknowledgments}
This work was supported by a research grant (VIL”52303”) from Villum Fonden.
The third author is very grateful for the hospitality of the Algebra group at DTU, he was visiting DTU during the development of this research in October 2024.
The research of the third author was also supported by the project ``COMPACT'' of the University of Campania ``Luigi Vanvitelli'' and was partially supported by the Italian National Group for Algebraic and Geometric Structures and their Applications (GNSAGA - INdAM).

\bibliographystyle{ieeetr}
\bibliography{ourbib}
\clearpage

\begin{appendix}
\section{Appendix}\label{sec:apen}
\begin{proof}[Proof of Claim~\ref{cl:est}]
\clest*
In this proof we will repeatedly use the following asymptotic estimate.
\begin{align} \label{eq:summ}
    \sum_{i=0}^{m} T^i \sim \frac{1}{m+1} T^{m+1} \quad \textnormal{as $m \to \infty$.}
\end{align}
We prove all the asymptotic estimates separately, but will use the following straightforward expansion in all of them.
\begin{align*}
    \frac{\binom{2k-1-2\ell}{s-\ell}}{\binom{2k-1}{s}} = \frac{\prod_{i=0}^{\ell-1}(s-i) \prod_{i=1}^{\ell}(2k-s-i)}{\prod_{i=1}^{2\ell}(2k-i)} \quad \textnormal{ for all $\ell \in [s]$.}
\end{align*}
\begin{itemize}
    \item[(i)] We have
    \begin{align*}
        \sum_{s=13}^{2k-3}\frac{s(s-1) (2k-s-1)(2k-s-2)}{\prod_{i=1}^{4}(2k-i)} &= \sum_{s=13}^{2k-3}\frac{{s(s-1)(4k^2-4ks-6k+s^2+3s+2)}}{\prod_{i=1}^{4}(2k-i)} \\
        &\sim \frac{1}{2^4k^4}\sum_{s=13}^{2k-3}{(4k^2s^2-4ks^3+s^4)} \\
        &= \frac{4k^2\sum_{s=13}^{2k-3}s^2-4k\sum_{s=13}^{2k-3}s^3+\sum_{s=13}^{2k-3}s^4}{(2k)^4} \\
        &\sim \frac{4k^2(2k)^3/3-4k(2k)^4/4+(2k)^5/5}{2^4k^4} \\ 
        &= \frac{k}{15}
    \end{align*}
    as $k \to \infty$, where in the last asymptotic estimate we used~\eqref{eq:summ}.
    
\item[(ii)] We have
\begin{align*}
        \sum_{s=13}^{2k-3}\frac{\prod_{i=0}^{2}(s-i) \prod_{i=1}^{3}(2k-s-i)}{\prod_{i=1}^{6}(2k-i)} &= \sum_{s=13}^{2k-3}\frac{{\prod_{i=0}^{2}(s-i)f(k,s)}}{\prod_{i=1}^{6}(2k-i)} \\
        &\sim \sum_{s=13}^{2k-3}\frac{{s^3f(k,s)}}{(2k)^6}
\end{align*}
where 
\begin{align*}
    f(k,s)=8k^3 - 12k^2 s - 36k^2 + 6ks^2 + 36ks + 44k - s^3 - 6s^2 - 11s - 6.
\end{align*}
Asymptotically, as $k \to \infty$, the terms of degree $3$ of $f(k,s)$ are preponderant, and so we obtain
\begin{align*}
    \sum_{s=13}^{2k-3}\frac{{s^3f(k,s)}}{(2k)^6} &\sim \sum_{s=13}^{2k-3}\frac{{s^3(8k^3 - 12k^2s +6ks^2 - s^3)}}{(2k)^6} \\
    &\sim \frac{{8k^3(2k)^4/4 - 12k^2(2k)^5/5 + 6k(2k)^6/6 - (2k)^7/7}}{(2k)^6} \\
    &= \frac{k}{70}
\end{align*}
where again we used~\eqref{eq:summ}.
\item[(iii)]   We have
\begin{align*}
       \sum_{s=13}^{2k-3}\frac{\prod_{i=0}^{3}(s-i) \prod_{i=1}^{4}(2k-s-i)}{\prod_{i=1}^{8}(2k-i)} &= \sum_{s=13}^{2k-3}\frac{{\prod_{i=0}^{3}(s-i)f(k,s)}}{\prod_{i=1}^{8}(2k-i)} \\
       &\sim \sum_{s=13}^{2k-3}\frac{{s^4f(k,s)}}{(2k)^8}
    \end{align*}
    as $k \to \infty$ where
    \begin{align*}
        f(k,s) =\;   &16k^4 - 32k^3 s - 160k^3 + 24k^2 s^2 + 192k^2 s + 292k^2 - 8k s^3 \\
        &- 96k s^2 - 268k s - 208k + s^4 + 10s^3 + 35s^2 + 50s + 24.
    \end{align*}
    Again, only the terms of degree $4$ play a role asymptotically and so we obtain
    \begin{align*}
        \sum_{s=13}^{2k-3}\frac{{s^4f(k,s)}}{(2k)^8} &\sim \sum_{s=13}^{2k-3}\frac{{s^4(16k^4 - 32k^3 s  + 24k^2 s^2 - 8k s^3 +s^4)}}{(2k)^8} \\
        &\sim \frac{{16k^4(2k)^5/5 - 32k^3(2k)^6/6  + 24k^2 (2k)^7/7 - 8k(2k)^8/8 +(2k)^9/9}}{(2k)^8} \\
        &= \frac{k}{315}
    \end{align*}
    as $k \to \infty$.
\item[(iv)]   We have
\begin{align*}
        \sum_{s=13}^{2k-3}\frac{\prod_{i=0}^{4}(s-i) \prod_{i=1}^{5}(2k-s-i)}{\prod_{i=1}^{10}(2k-i)} &= \sum_{s=13}^{2k-3}\frac{{\prod_{i=0}^{4}(s-i)f(k,s)}}{\prod_{i=1}^{10}(2k-i)} \\
        &\sim \sum_{s=13}^{2k-3}\frac{{s^5f(k,s)}}{(2k)^{10}}
        \end{align*}
        where
    \begin{align*}
        f(k,s) = \;  &32k^5 - 80k^4 s - 240k^4 + 80k^3 s^2 + 480k^3 s + 680k^3 - 40k^2 s^3 \\
        &- 360k^2 s^2 - 1020k^2 s - 900k^2 + 10k s^4 + 120k s^3 + 510k s^2 \\
        &+ 900k s + 548k - s^5 - 15s^4 - 85s^3 - 225s^2 - 274s - 120.
    \end{align*}
    Only the terms of degree $5$ play a role asymptotically and so we obtain
         \begin{align*}
        \sum_{s=13}^{2k-3}\frac{{s^5f(k,s)}}{(2k)^{10}} &\sim \sum_{s=13}^{2k-3}\frac{{s^5(32k^5 - 80k^4 s + 80k^3 s^2  - 40k^2 s^3 + 10k s^4- s^5)}}{(2k)^{10}} \\
        &\sim \frac{k}{1386}
    \end{align*}
\item[(v)]   We have
\begin{align*}
        \sum_{s=13}^{2k-3}\frac{\prod_{i=0}^{5}(s-i) \prod_{i=1}^{6}(2k-s-i)}{\prod_{i=1}^{12}(2k-i)} &= \sum_{s=13}^{2k-3}\frac{{\prod_{i=0}^{5}(s-i)f(k,s)}}{\prod_{i=1}^{12}(2k-i)} \\
        &\sim \sum_{s=13}^{2k-3}\frac{{s^6f(k,s)}}{(2k)^{12}}
        \end{align*}
        where
    \begin{align*}
        f(k,s) = \; &64k^6 - 192k^5 s - 672k^5 + 240k^4 s^2 + 1680k^4 s + 2800k^4 - 160k^3 s^3  \\
        &- 1680k^3 s^2 - 5600k^3 s - 5880k^3 + 60k^2 s^4 + 840k^2 s^3 + 4200k^2 s^2\\
        & + 8820k^2 s + 6496k^2 - 12k s^5 - 120k s^4 - 1400k s^3 - 4410k s^2 - 6496k s\\
        & - 3528k + s^6 + 21s^5 + 175s^4 + 735s^3 + 1624s^2 + 1764s + 720.
    \end{align*}
     Only the terms of degree $6$ play a role asymptotically and so we obtain
         \begin{align*}
        \sum_{s=13}^{2k-3}\frac{{s^6f(k,s)}}{(2k)^{12}} &\sim \sum_{s=13}^{2k-3}\frac{{s^6(64k^6 - 192k^5 s + 240k^4 s^2 - 160k^3 s^3  + 60k^2 s^4 -12ks^5+ s^6 )}}{(2k)^{12}} \\
        &\sim \frac{k}{6006}
    \end{align*}
    as $k \to \infty$.
\item[(vi)]   We have
\begin{align*}
        \sum_{s=13}^{2k-3}\frac{\prod_{i=0}^{6}(s-i) \prod_{i=1}^{7}(2k-s-i)}{\prod_{i=1}^{14}(2k-i)} &= \sum_{s=13}^{2k-3}\frac{{\prod_{i=0}^{6}(s-i)f(k,s)}}{\prod_{i=1}^{14}(2k-i)}
        \\
        &\sim \sum_{s=13}^{2k-3}\frac{{s^7f(k,s)}}{(2k)^{14}}
        \end{align*}
        where
    \begin{align*}
        f(k,s) =\; &128k^7 - 448k^6 s - 1792k^6 + 672k^5 s^2 + 5376k^5 s + 10304k^5 - 560k^4 s^3 \\ 
        &- 6720k^4 s^2  - 25760k^4 s - 31360k^4
+ 280k^3 s^4 + 4480k^3 s^3 + 25760k^3 s^2 \\ 
        &+ 62720k^3 s + 54152k^3 - 84k^2 s^5 - 1680k^2 s^4 - 12880k^2 s^3 - 47040k^2 s^2 \\ 
        &- 81228k^2 s - 52528k^2 + 14k s^6 + 336k s^5 + 3220k s^4 + 15680k s^3 \\ 
        &+ 40614k s^2 + 52528k s + 26136k - s^7 - 28s^6 - 322s^5 - 1960s^4 \\ 
        &- 6769s^3 - 13132s^2 - 13068s - 5040.
    \end{align*}
    Only the terms of degree $7$ play a role asymptotically and so we obtain
         \begin{align*}
        \sum_{s=13}^{2k-3}\frac{{s^7f(k,s)}}{(2k)^{14}} \sim \sum_{s=13}^{2k-3}\frac{{s^7g(k,s)}}{(2k)^{14}} \sim \frac{k}{25740}
    \end{align*}
    as $k \to \infty$ where 
    \begin{align*}
        g(k,s) = 128k^7 - 448k^6 s + 672k^5 s^2 - 560k^4 s^3 + 280k^3 s^4  - 84k^2 s^5+ 14k s^6 - s^7.
    \end{align*}
\item[(vii)]   We have
\begin{align*}
        \sum_{s=13}^{2k-3}\frac{\prod_{i=0}^{7}(s-i) \prod_{i=1}^{8}(2k-s-i)}{\prod_{i=1}^{16}(2k-i)} &= \sum_{s=13}^{2k-3}\frac{{\prod_{i=0}^{7}(s-i)f(k,s)}}{\prod_{i=1}^{16}(2k-i)}
        \\
        &\sim \sum_{s=13}^{2k-3}\frac{{s^8f(k,s)}}{(2k)^{16}}
        \end{align*}
        where
    \begin{align*}
        f(k,s) = \; &256k^8 - 1024k^7 s - 4608k^7 + 1792k^6 s^2 + 16128k^6 s + 34944k^6 - 1792k^5 s^3 \\
&- 24192k^5 s^2 - 104832k^5 s - 145152k^5 + 1120k^4 s^4 + 20160k^4 s^3 + 131040k^4 s^2 \\
&+ 362880k^4 s + 359184k^4 - 448k^3 s^5 - 10080k^3 s^4 - 87360k^3 s^3 - 362880k^3 s^2 \\
&- 718368k^3 s - 538272k^3 + 112k^2 s^6 + 3024k^2 s^5 + 32760k^2 s^4 + 181440k^2 s^3 \\
&+ 538776k^2 s^2 + 807408k^2 s + 472496k^2 - 16k s^7 - 504k s^6 - 6552k s^5 \\
&- 45360k s^4 - 179592k s^3 - 403704k s^2 - 472496k s - 219168k + s^8 + 36 s^7 \\
&+ 546 s^6 + 4536 s^5 + 22449 s^4 + 67284 s^3 + 118124 s^2 + 109584 s + 40320.
    \end{align*}
    Only the terms of degree $8$ play a role asymptotically and so we obtain
         \begin{align*}
        \sum_{s=13}^{2k-3}\frac{{s^8f(k,s)}}{(2k)^{16}} \sim \sum_{s=13}^{2k-3}\frac{{s^8g(k,s)}}{(2k)^{16}} \sim \frac{k}{109395}
    \end{align*}
    as $k \to \infty$ where 
    \begin{align*}
        g(k,s) =\; &256k^8 - 1024k^7 s + 1792k^6 s^2 - 1792k^5 s^3  + 1120k^4 s^4 - 448k^3 s^5 \\
        &+ 112k^2 s^6 - 16k s^7  + s^8.
    \end{align*}
\item[(viii)]   We have
\begin{align*}
        \sum_{s=13}^{2k-3}\frac{\prod_{i=0}^{8}(s-i) \prod_{i=1}^{9}(2k-s-i)}{\prod_{i=1}^{18}(2k-i)} &= \sum_{s=13}^{2k-3}\frac{{\prod_{i=0}^{8}(s-i)f(k,s)}}{\prod_{i=1}^{18}(2k-i)}
        \\
        &\sim \sum_{s=13}^{2k-3}\frac{{s^9f(k,s)}}{(2k)^{18}}
        \end{align*}
        where
    \begin{align*}
        f(k,s) =\; &512k^9 - 2304k^8 s - 11520k^8 + 4608k^7 s^2 + 46080k^7 s + 111360k^7 - 5376k^6 s^3 \\
&- 80640k^6 s^2 - 389760k^6 s - 604800k^6 + 4032k^5 s^4 + 80640k^5 s^3 + 584640k^5 s^2 \\
&+ 1814400k^5 s + 2024736k^5 - 2016k^4 s^5 - 50400k^4 s^4 - 487200k^4 s^3 \\
&- 2268000k^4 s^2- 5061840k^4 s - 4309200k^4 + 672k^3 s^6 + 20160k^3 s^5 \\
&+ 243600k^3 s^4 + 1512000k^3 s^3 + 5061840k^3 s^2 + 8618400k^3 s+ 5789440k^3 \\
&- 144k^2 s^7 - 5040k^2 s^6 - 73080k^2 s^5 - 567000k^2 s^4 - 2530920k^2 s^3 \\
&- 6463800k^2 s^2 - 8684160k^2 s- 4690800k^2 + 18k s^8 + 720k s^7 + 12180k s^6 \\
&+ 113400k s^5 + 632730k s^4 + 2154600k s^3 + 4342080k s^2 + 4690800k s \\
&+ 2053152k - s^9 - 45s^8 - 870s^7 - 9450s^6 - 63273s^5 - 269325s^4 \\
&- 723680s^3 - 1172700s^2 - 1026576s - 362880.
    \end{align*}
     Only the terms of degree $9$ play a role asymptotically and so we obtain
         \begin{align*}
        \sum_{s=13}^{2k-3}\frac{{s^9f(k,s)}}{(2k)^{18}} \sim \sum_{s=13}^{2k-3}\frac{{s^9g(k,s)}}{(2k)^{18}} \sim \frac{k}{461890}
    \end{align*}
    as $k \to \infty$ where 
    \begin{align*}
        g(k,s) =\; &512k^9 - 2304k^8 s + 4608k^7 s^2 - 5376k^6 s^3 + 4032k^5 s^4 \\
        &- 2016k^4 s^5  + 672k^3 s^6 - 144k^2 s^7 + 18k s^8 - s^9.
    \end{align*}

\item[(ix)]   We have
\begin{align*}
        \sum_{s=13}^{2k-3}\frac{\prod_{i=0}^{9}(s-i) \prod_{i=1}^{10}(2k-s-i)}{\prod_{i=1}^{20}(2k-i)} &= \sum_{s=13}^{2k-3}\frac{{\prod_{i=0}^{9}(s-i)f(k,s)}}{\prod_{i=1}^{20}(2k-i)}
         \\
        &\sim \sum_{s=13}^{2k-3}\frac{{s^{10}f(k,s)}}{(2k)^{20}}
        \end{align*}
        where
    \begin{align*}
        f(k,s) =\; &1024 k^{10} - 5120 k^9 s - 28160 k^9 + 11520 k^8 s^2 + 126720 k^8 s + 337920 k^8 \\
        &- 15360 k^7 s^3 - 253440 k^7 s^2 - 1351680 k^7 s - 2323200 k^7 + 13440 k^6 s^4 \\
        &+ 295680 k^6 s^3 + 2365440 k^6 s^2 + 8131200 k^6 s + 10097472 k^6 - 8064 k^5 s^5 \\
        &- 221760 k^5 s^4 - 2365440 k^5 s^3 - 12196800 k^5 s^2 - 30292416 k^5 s \\
        &- 28865760 k^5  + 3360 k^4 s^6 + 118080 k^4 s^5 + 1478400 k^4 s^4 + 10166400 k^4 s^3 \\
        &+ 37865520 k^4 s^2 + 72164400 k^4 s + 54670880 k^4 - 960 k^3 s^7 - 36960 k^3 s^6 \\
        &- 591360 k^3 s^5 - 5082000 k^3 s^4 - 25243680 k^3 s^3 - 72164400 k^3 s^2 \\
        &- 109347160 k^3 s - 67276000 k^3 + 180 k^2 s^8 + 7920 k^2 s^7 + 147840 k^2 s^6 \\
        &+ 1524600 k^2 s^5 + 9466380 k^2 s^4 + 36082200 k^2 s^3 + 82006320 k^2 s^2 \\
        &+ 100914000 k^2 s + 51014304 k^2 - 20 k s^9 - 990 k s^8 - 21120 k s^7 \\
        &- 254100 k s^6 - 1893276 k s^5 - 9020550 k s^4 - 27335440 k s^3 - 50457700 k s^2 \\
        &- 51014304 k s - 21257280 k  + s^{10} + 55 s^{9} + 1230 s^{8} + 18150 s^{7} \\
        &+ 157773 s^{6} + 902055 s^{5} + 3416930 s^{4} + 8409500 s^{3} + 12753576 s^{2} + 10628640 s \\
        &+ 3628800.
    \end{align*}
     Only the terms of degree $10$ play a role asymptotically and so we obtain
         \begin{align*}
        \sum_{s=13}^{2k-3}\frac{{s^{10}f(k,s)}}{(2k)^{20}} \sim \sum_{s=13}^{2k-3}\frac{{s^{10}g(k,s)}}{(2k)^{20}} \sim \frac{k}{1939938}
    \end{align*}
    as $k \to \infty$ where 
    \begin{align*}
        g(k,s) =\;  &1024 k^{10} - 5120 k^9 s  + 11520 k^8 s^2 - 15360 k^7 s^3  + 13440 k^6 s^4  \\
        &- 8064 k^5 s^5 + 3360 k^4 s^6  - 960 k^3 s^7  + 180 k^2 s^8  - 20 k s^9  + s^{10}.
    \end{align*}

\item[(x)]   We have
\begin{align*}
        \sum_{s=13}^{2k-3}\frac{\prod_{i=0}^{10}(s-i) \prod_{i=1}^{11}(2k-s-i)}{\prod_{i=1}^{22}(2k-i)} &= \sum_{s=13}^{2k-3}\frac{{\prod_{i=0}^{10}(s-i)f(k,s)}}{\prod_{i=1}^{22}(2k-i)}  \\
        &\sim \sum_{s=13}^{2k-3}\frac{{s^{11}f(k,s)}}{(2k)^{22}}
        \end{align*}
        where
    \begin{align*}
        f(k,s) = \; &2048 k^{11} - 11264 k^{10} s - 67584 k^{10} + 28160 k^9 s^2 + 337920 k^9 s + 985600 k^9 \\
        &- 42240 k^8 s^3 - 760320 k^8 s^2 - 4435200 k^8 s - 8363520 k^8 + 42240 k^7 s^4 \\
        &+ 1013760 k^7 s^3 + 8870400 k^7 s^2 + 33454080 k^7 s + 45750144 k^7 - 29568 k^6 s^5 \\
        &- 887040 k^6 s^4 - 10348800 k^6 s^3 - 58544640 k^6 s^2 - 160125504 k^6 s - 168803712 k^6 \\
        &+ 14784 k^5 s^6 + 532224 k^5 s^5 + 7761600 k^5 s^4 + 58544640 k^5 s^3 \\
        &+ 240188256 k^5 s^2 + 506411136 k^5 s + 426865120 k^5 - 5280 k^4 s^7 - 221760 k^4 s^6 \\
        &- 3880880 k^4 s^5 - 36590400 k^4 s^4 - 200156880 k^4 s^3 - 633013920 k^4 s^2 \\
        &- 1067162800 k^4 s - 735931680 k^4 + 1320 k^3 s^8 + 63360 k^3 s^7+1293600 k^3 s^6 \\
        &+ 14636160 k^3 s^5 +100078440k^3s^4+ 422009280 k^3 s^2 + 1067162800 k^3 s \\
        &+ 1471863360 k^3 + 842064608 k^3 - 220 k^2 s^9 - 11880 k^2 s^8 - 277200 k^2 s^7 \\
        &- 3659040 k^2 s^6 - 30023532 k^2 s^5 - 158253480 k^2 s^4 - 533581400 k^2 s^3 \\
        &- 1103897520 k^2 s^2 - 1263096912 k^2 s - 603671904 k^2 + 22 k s^{10} + 1320 k s^9 \\
        &+ 34650 k s^8 + 522720 k s^7 + 5003922 k s^6 + 31650696 k s^5 + 133395350 k s^4 \\
        &+ 367965840 k s^3 + 631548456 k s^2 + 603671904 k s + 241087680k - s^{11}- 66s^{10} \\
        &- 1925 s^9 - 32670 s^8 - 357423 s^7 - 2637558 s^6 - 13339535 s^5 - 45995730 s^4 \\
        &- 105258076 s^3 - 150917976 s^2 - 120543840 s - 39916800.
    \end{align*}
     Only the terms of degree $11$ play a role asymptotically and so we obtain
         \begin{align*}
        \sum_{s=13}^{2k-3}\frac{{s^{11}f(k,s)}}{(2k)^{22}} \sim \sum_{s=13}^{2k-3}\frac{{s^{11}g(k,s)}}{(2k)^{22}} \sim \frac{k}{8112468}
    \end{align*}
    as $k \to \infty$ where 
    \begin{align*}
        g(k,s) =\; &2048 k^{11} - 11264 k^{10} s + 28160 k^9 s^2 - 42240 k^8 s^3 + 42240 k^7 s^4 - 29568 k^6 s^5 \\
        &+ 14784 k^5 s^6  - 5280 k^4 s^7  + 1320 k^3 s^8  - 220 k^2 s^9 + 22 k s^{10}  - s^{11}. \quad \qedhere
    \end{align*}
\end{itemize}
\end{proof}

\end{appendix}

\end{document}